\documentclass[]{interact}

\usepackage{epstopdf}
\usepackage[caption=false]{subfig}

\usepackage[numbers,sort&compress]{natbib}
\bibpunct[, ]{[}{]}{,}{n}{,}{,}
\makeatletter
\def\NAT@def@citea{\def\@citea{\NAT@separator}}
\makeatother

\usepackage{color}
\usepackage{mathrsfs}
\usepackage{bbm}
\usepackage{comment}

\theoremstyle{plain}
\newtheorem{theorem}{Theorem}[section]
\newtheorem{lemma}[theorem]{Lemma}
\newtheorem{corollary}[theorem]{Corollary}
\newtheorem{proposition}[theorem]{Proposition}

\theoremstyle{definition}
\newtheorem{definition}[theorem]{Definition}
\newtheorem{example}[theorem]{Example}

\theoremstyle{remark}
\newtheorem{remark}{Remark}

\newcommand\rank[1]{\operatorname{rank}({#1})}

\newcommand{\identity}{\mathbbm{1}}
\newcommand{\tr}{\operatorname{tr}}
\newcommand{\diag}{\operatorname{diag}}
\newcommand{\id}{\operatorname{id}{\!}}
\newcommand{\conj}{*}
\newcommand{\unitvector}{e}

\begin{document}

\title{Strict Positivity and $D$-Majorization}

\author{
\name{Frederik vom Ende\textsuperscript{a,b}\thanks{CONTACT Frederik vom Ende. Email: frederik.vom-ende@tum.de}}
\affil{\textsuperscript{a}Department of Chemistry, Lichtenbergstra{\ss}e 4, 85747 Garching, Germany\\
\textsuperscript{b}Munich Centre for Quantum Science and Technology (MCQST),  Schellingstra{\ss}e~4, 80799~M{\"u}nchen, Germany}
}

\maketitle

\begin{abstract}
Motivated by quantum thermodynamics we first investigate the notion of strict positivity, that is, linear maps which map positive definite states to something positive definite again. We show that strict positivity is decided by the action on any full-rank state, and that the image of non-strictly positive maps lives inside a lower-dimensional subalgebra. This implies that the distance of such maps to the identity channel is lower bounded by one. 

The notion of strict positivity comes in handy when generalizing the majorization ordering on real vectors with respect to a positive vector $d$ to majorization on square matrices with respect to a positive definite matrix $D$. For the two-dimensional case we give a characterization of this ordering via finitely many trace norm inequalities and, moreover, investigate some of its order properties. In particular it admits a unique minimal and a maximal element. The latter is unique as well if and only if minimal eigenvalue of $D$ has multiplicity one.
\end{abstract}

\begin{amscode}
15A45, 15B48, 47B65, 81P47
\end{amscode}
\section{Introduction}
A fundamental aspect of resource theories is finding conditions which characterize state-transfers via ``allowed'' operations. In quantum thermodynamics, for example, one usually asks whether a state can be generated from an initial state via a Gibbs-preserving quantum channel, that is, a
completely positive and trace-preserving 
map which preserves the Gibbs state of the system \cite{Brandao15,DallArno20,Gour15,Horodecki13}. 
As Gibbs states are of the form $e^{-\beta H}/\tr(e^{-\beta H})$ for some system's Hamiltonian $H$---usually a hermitian $n\times n$ matrix---
and some inverse temperature $\beta>0$, these states in particular are of full rank. Actually every full-rank state $D$ is the Gibbs state of some system by simply choosing $\beta=1$ and $H=\ln(D^{-1})$. Therefore gaining a better understanding of the set of all channels with a given full-rank fixed point and its geometry, properties, etc.~would greatly benefit the aforementioned state-conversion problem. Channels with a full-rank fixed point, sometimes called \textit{faithful} \cite{Albert19}, are characterized by not containing a decaying subspace (under the asymptotic projection). Also conserved quantities of such channels commute with all Kraus operators up to a phase \cite[Prop.~1]{Albert19}. 
Topics related to faithful channels and fixed-point analysis of {\sc cptp} maps are (mean-)ergodic channels \cite{Burgarth13,Ohya81}, irreducible channels \cite{Davies70,Sanz10}, zero-error \cite{BlumeKohout10}, \cite[Ch.~4]{Gupta15} and relaxation properties of discrete-time \cite{Cirillo15} and continuous-time \cite{Frigerio77,Spohn76,Spohn77} Markovian systems.

As a notable special case if the initial and the final state commute with the Hamiltonian then one is in the classical realm which, mathematically 
speaking, is handled by majorization relative to a positive vector $d\in\mathbb R^n$ as introduced by Veinott \cite{Veinott71} 
and (in the quantum regime) Ruch, Schranner and Seligman \cite{Ruch78} (cf.~\cite[Appendix E]{Brandao15}). Given such positive $d$ some vector 
$x$ is said to $d$-majorize $y$, denoted by $x\prec_d y$, if there exists a column-stochastic\footnote{A matrix 
$A\in\mathbb R^{n\times n}$ is said to be \textit{column stochastic} if all its entries are non-negative and 
$\sum\nolimits_{i=1}^n A_{ij}=1$ for all $j=1,\ldots,n$, i.e.~the entries of each column sums up to one.\label{footnote:column_stoch}} matrix $A$ with 
$Ad=d$ and $x=Ay$. 
A variety of characterizations of $\prec_d$ are collected at the start of Chapter \ref{sec:matrix_d_maj}. In particular the above state-conversion problem in this classical case reduces to $n$ vector-$1$-norm inequalities (where $n$ is the dimension of the system) and the vectors $d$-majorized by some initial vector form a convex polytope with at most $n!$ extreme points. In addition, these extreme points can be easily computed analytically \cite{vomEnde19polytope}. For more information on classical as well as more general notions of majorization as well as its applications we refer to \cite{MarshallOlkin}. 

When studying channels with a full-rank fixed point one finds that those belong to the larger class of linear maps which preserve positive definiteness, sometimes called \textit{strictly positive} maps. While this is a rather large class of maps it turns out that some results regarding positive and strictly positive maps will be useful tools when generalizing $d$-majorization from vectors to matrices. 
This is why---after a short review of quantum channels in Chapter \ref{ch:prelim}---Section \ref{sec:3_1} features the definition of strict positivity, its properties and how it ``interacts'' with complete positivity. In Section \ref{sec:nonstrpos} we study positive maps which are not strictly positive and prove that their image lives inside a lower-dimensional subalgebra. Also we discuss some physical applications of these results such as characterizing qubit channels which are not strictly positive (Corollary \ref{coro1}) or proving that all Markovian quantum processes are strictly positive at all times (Remark \ref{rem_markov}). 

In Chapter \ref{sec:matrix_d_maj} after recapping $d$-majorization on real vectors 
we present the motivation for our definition of $D$-majorization on matrices---denoted by $\prec_D$---and
link some special cases to the vector situation. Following up in Section \ref{subsec:1b} we characterize $\prec_D$ in the qubit case via the trace norm and elaborate on why this is a challenging task when going 
beyond two dimensions. Finally Section \ref{subsec:2} deals with characterizing minimal and maximal 
elements of this preorder (Thm.~\ref{d_matrix_minmax}) as well as the set of all matrices $D$-majorized by 
some initial state, and its properties (Thm.~\ref{lemma_R_d_closed} \& Remark \ref{rem_md_cont}).
 
\section{Recap: Quantum States and Quantum Channels}\label{ch:prelim}

Before talking about strict positivity as well as majorization in the context of 
matrices we need some well-known results from quantum information theory. For the purpose of this paper be aware of the following notions and notations.
\begin{itemize}
\item The tensor product---more specifically the Kronecker product---is denoted by $\otimes$.\smallskip
\item We take the inner product $\langle \cdot,\cdot\rangle$ to be antilinear in the first and linear in the second component. Given some vectors $a,b\in\mathbb C^n$ the bra-ket notation $|b\rangle\langle a|$ refers to the linear map $x\mapsto \langle a,x\rangle b$ on $\mathbb C^n$.\smallskip
\item The trace norm for any $A\in \mathbb C^{n\times n}$ is given by $\|A\|_1=\tr(\sqrt{A^\conj A})$.\smallskip
\item Some $A\in \mathbb C^{n\times n}$ is called positive semi-definite (positive definite), denoted by $A\geq 0$ ($>0$), if $\langle x,Ax\rangle\geq 0$ ($>0$) for all non-zero vectors $x\in\mathbb C^n$. Because the considered field is $\mathbb C$, $\langle x,Ax\rangle\geq 0$ in particular implies that $A$ is hermitian due to the polarization identity---over the reals, however, symmetry of $A$ would have to be required explicitly. \smallskip
\item The collection of all \textit{states} (or \textit{density matrices}) of size $n$ is defined to be
$
\mathbb D(\mathbb C^n):=\lbrace \rho\in\mathbb C^{n\times n}\,|\,\rho\geq 0\text{ and }\tr(\rho)=1\rbrace\,.
$
\end{itemize}

We start with a well-known result.
\begin{lemma}\label{lemma_pure_extreme}
For all $n\in\mathbb N$ the set of all states $\mathbb D(\mathbb C^n)$ is convex and compact, and the rank-one projections (called pure states) are precisely its extreme points.
\end{lemma}
\begin{proof}
Convexity and boundedness are evident. For closedness consider some sequence in $\mathbb D(\mathbb C^n)$ which converges in trace norm to some $A\in\mathbb C^{n\times n}$. Then in particular this sequence converges in trace ($|\tr(\cdot)|\leq\|\cdot\|_1$) and weakly ($|\langle x,(\cdot)x\rangle|\leq \|\cdot\|\|x\|^2\leq \|\cdot\|_1\|x\|^2$ for all $x\in\mathbb C^n$) hence $A\in\mathbb D(\mathbb C^n)$. 
Finally the statement regarding the extreme points is for example shown in \cite[Thm.~2.3]{Holevo12}. 
\end{proof}
A linear map $T:\mathbb C^{n\times n}\to\mathbb C^{k\times k}$  where, here and henceforth, $k,n\in\mathbb N$ are arbitrary is said to be positivity-preserving (for short: \textit{positive}) if $T(A)\geq 0$ for all $A\geq 0$. It furthermore is said to be \textit{completely positive}
if $T\otimes\operatorname{id}_{m}:\mathbb C^{nm\times nm}\to \mathbb C^{km\times km}$
is positive for all $m\in\mathbb N$. Obviously if domain and co-domain are fixed the (completely) positive maps form a semigroup. 
As we are in finite dimensions, complete positivity can be characterized as follows \cite{Choi75}:
\begin{lemma}\label{lemma_choi_matrix}
Let linear $T:\mathbb C^{n\times n}\to\mathbb C^{k\times k}$ be given. The following are equivalent.
\begin{itemize}
\item[(i)] $T$ is completely positive.\smallskip
\item[(ii)] The Choi matrix $
C(T):=\big(T(|e_j\rangle\langle e_k|)\big)_{j,k=1}^n$ of $T$ is positive semi-definite.\smallskip
\item[(iii)] There exist $\{K_i\}_{i\in I}\subset\mathbb C^{n\times k}$ with $|I|\leq nk$---called Kraus operators---such that $T(A)=\sum_{i\in I}K_i^\conj AK_i$ for all $A\in\mathbb C^{n\times n}$.
\end{itemize}
\end{lemma}
With this a \textit{quantum channel} is a linear, completely positive and trace-preserving map
$T:\mathbb C^{n\times n}\to\mathbb C^{k\times k}$, also called {\sc cptp} map. Moreover one defines
\begin{align*}
Q(n,k):=\lbrace T:\mathbb C^{n\times n}\to\mathbb C^{k\times k}\, |\, T\text{ is quantum channel\,}\rbrace
\end{align*}
and $Q(n):=Q(n,n)$. The set of channels with a common fixed point $X\in\mathbb D(\mathbb C^n)$ is denoted by
$
Q_X(n):=\lbrace T\in Q(n)\,|\,T(X)=X\rbrace
$. One finds\footnote{
The positive trace-preserving maps are precisely those linear maps $T$ which satisfy $T(\mathbb D(\mathbb C^n))\subseteq\mathbb D(\mathbb C^n)$---hence this also holds for every quantum channel. As the states form a convex and compact set (Lemma \ref{lemma_pure_extreme}) by the Brouwer fixed-point theorem \cite{Brouwer11} every such $T$ has a fixed point in $\mathbb D(\mathbb C^n)$.
}
$Q(n)=\bigcup_{X\in\mathbb D(\mathbb C^n)}Q_X(n)$. 
We list the following important properties of quantum channels.
\begin{lemma}\label{lemma_qc_op_bound}
Every map $T:(\mathbb C^{n\times n},\|\cdot\|_1)\to (\mathbb C^{k\times k},\|\cdot\|_1)$ which is linear, positive and trace-preserving has operator norm $\|T\|= 1$.  
\end{lemma}
\begin{proof}
On the one hand $\|T\|\leq 1$ \cite[Thm.~2.1]{Wolf06}. On the other hand, $T$ being positive and trace-preserving implies that it is a linear isometry on the positive semi-definite matrices as there trace norm and trace coincide---hence $\|T\|\geq 1$.
\end{proof}

\begin{lemma}\label{lemma_convex_subsemigroup}
The set $Q(n)$, as well as $Q_X(n)$ for arbitrary $X\in\mathbb D(\mathbb C^{n})$, forms a convex and compact semigroup with identity element $\operatorname{id}_n$.
\end{lemma}
\begin{proof}
Convexity and semigroup properties of $Q(n)$ are well known \cite[Ch.~4.3]{Heinosaari12}. Boundedness is due to Lemma \ref{lemma_qc_op_bound}. Closedness was shown, e.g., in \cite[Thm.~1]{vE_dirr_semigroups}. One readily verifies that this argument still works when replacing $Q(n)$ by $Q_X(n)$. \end{proof}

\begin{proposition}\label{prop_exist_channel}
Let $A,B\in\mathbb C^{n\times n}$ be hermitian. Then the following are equivalent.
\begin{itemize}
\item[(i)] $\tr(A)=\tr(B)$ and $\|A\|_1\leq\|B\|_1$. \smallskip
\item[(ii)] There exists $T\in Q(n)$ such that $T(B)=A$.  \smallskip
\item[(iii)] There exists $T:\mathbb C^{n\times n}\to \mathbb C^{n\times n}$ linear, positive and trace-preserving such that $T(B)=A$. 
\end{itemize}
Moreover, if (i) holds and $0$ is an eigenvalue of $B$ then there exists $\psi\in\mathbb C^n$ with $\langle\psi,\psi\rangle=1$ such that $T(|\psi\rangle\langle\psi|)$ can be chosen arbitrarily from $\mathbb D(\mathbb C^n)$. 
\end{proposition}
\begin{proof}[Proof idea]
First, one constructively proves an analogous statement for real vectors and column-stochastic matrices which then can be directly lifted to the matrix case. A complete proof can be found in Appendix \ref{app_a0}. 
\end{proof}

\section{Strict Positivity}\label{sec:strictpos}

While the term ``strict positivity'' in the context of Perron-Frobenius theory refers to maps which send positive semi-definite operators to positive definite ones \cite{Farenick96,Gaubert17,Rahaman20} we want it to mean the following:

\begin{definition}[\cite{Bhatia07}, Ch.~2.2]
A linear map $T:\mathbb C^{n\times n}\to\mathbb C^{k\times k}$ is called strictly positive ({\sc sp}) if $T(X)>0$ whenever $X>0$. Moreover $T$ is called completely strictly positive ({\sc csp}) if $T\otimes\id_m$ is strictly positive for all $m\in\mathbb N$.
\end{definition}

This compares to usual positivity ({\sc p}) and complete positivity ({\sc cp}) as follows:
\begin{figure}[!htb]
\centering
\includegraphics[width=0.55\textwidth]{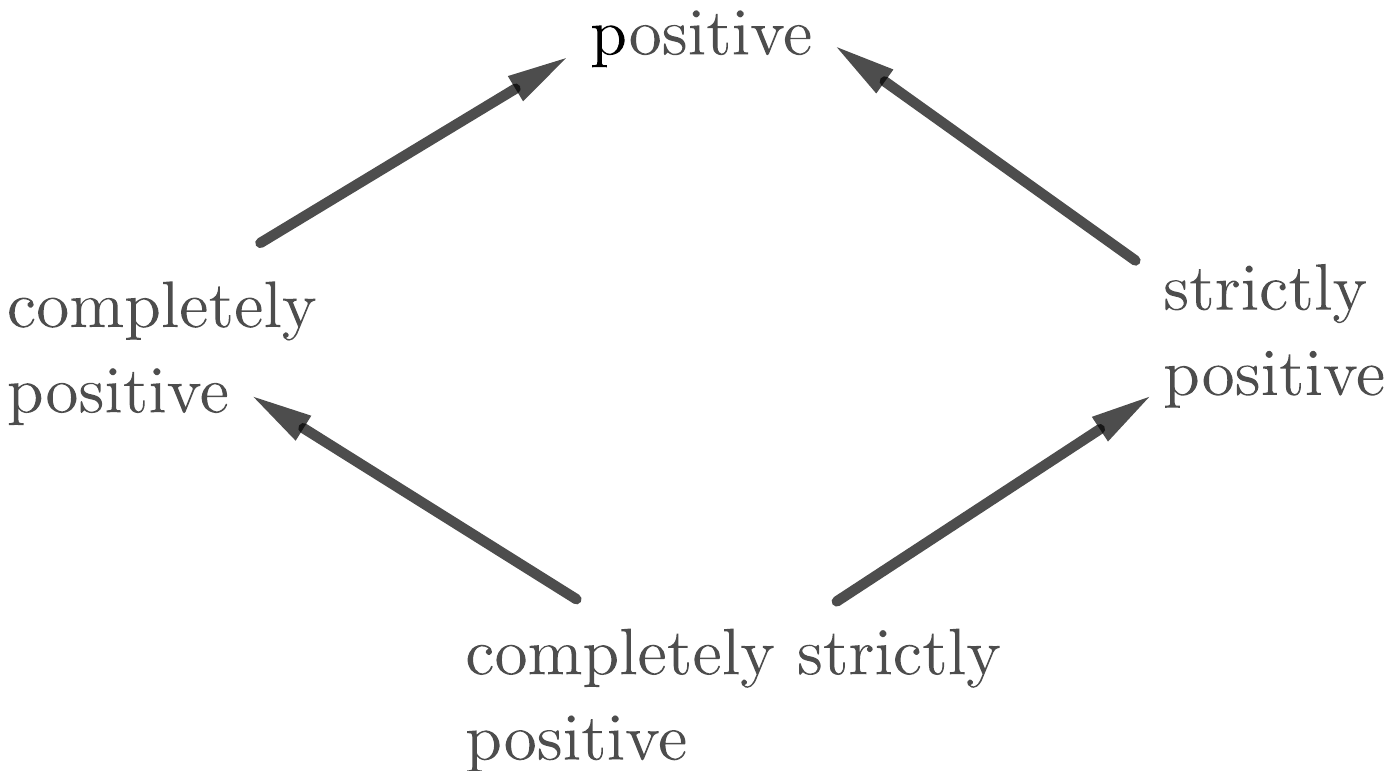}
\caption{Relation between {\sc p, cp, sp} and {\sc csp}. By definition {\sc csp} implies {\sc sp}. A simple continuity-type argument shows that {\sc sp} $\Rightarrow$ {\sc p} and {\sc csp} $\Rightarrow$ {\sc cp}. Moreover {\sc sp} and {\sc cp} are incomparable: the transposition map is obviously {\sc sp} but not {\sc cp} and, on the other hand, the trace projection $X\mapsto \tr(X)|\psi\rangle\langle\psi|$ for some pure state $\psi$ is a quantum channel \cite[Ex.~5.3]{Hayashi06} (hence {\sc cp}) but evidently not {\sc sp}, unless the dimension equals one.}
\label{fig1}
\end{figure}${}$

\subsection{Basic and Topological Properties}\label{sec:3_1}
Bhatia \cite{Bhatia07} observed that ``a positive linear map $\Phi$ is strictly positive if and only if $\Phi(\identity)>0$'' so strict positivity can be easily checked. This, however, turns out to be a mere corollary of the following stronger result: the image of positive linear maps admit a ``universal kernel'' completely characterized by any positive definite matrix.
\begin{proposition}\label{prop_positivity}
Let $T:\mathbb C^{n\times n}\to \mathbb C^{k\times k}$ linear and positive be given. For all \mbox{$X,Y,A\in\mathbb C^{n\times n}$}, $X,Y>0$ one has
\begin{equation}\label{eq:ker_T}
\operatorname{ker}(T(X))=\operatorname{ker}(T(Y))\subseteq\operatorname{ker}(T(A))\,.
\end{equation}
In particular the following are equivalent:
\begin{itemize}
\item[(i)] $T$ is strictly positive.\smallskip
\item[(ii)] $T(\identity)>0$\smallskip
\item[(iii)] There exists $X>0$ such that $T(X)>0$.\smallskip
\end{itemize}
\end{proposition}
\begin{proof}
Let $X\in\mathbb C^{n\times n}$ and a vector $\psi$ be given such that $X>0$ and $T(X)\psi=0$. Then for all $Y\in\mathbb C^{ n \times n }$ positive semi-definite one finds $\lambda\in\mathbb R$ such that\footnote{
E.g., choose $\lambda=\frac{y}{x}$ where $y$ is the largest eigenvalue of $Y$ and $x>0$ is the smallest eigenvalue of $X$.
}
$\lambda X-Y\geq 0$. But then $T(\lambda X-Y)\geq 0$ by positivity of $T$ so linearity shows
$$
0\leq \langle \psi,T(\lambda X-Y)\psi\rangle=-\langle \psi,T(Y)\psi\rangle\leq 0\,.
$$
Now $\|\sqrt{T(Y)}\psi\|^2=\langle \psi,T(Y)\psi\rangle=0$ is equivalent to $T(Y)\psi=0$ which shows $\ker(T(X))\subseteq\ker(T(Y))$. The case of a general $A\in\mathbb C^{n\times n}$ follows from the fact that every matrix can be written as a linear combination of four positive semi-definite matrices \cite[Coro.~4.2.4]{Kadison83} together with linearity of $T$. Finally if $Y>0$ then we can interchange the roles of $X,Y$ in the above argument to obtain $\ker(T(X))=\ker(T(Y))$.

For the second statement---while (i) $\Rightarrow$ (ii) $\Rightarrow$ (iii) is obvious---for (iii) $\Rightarrow$ (i) note that if $T(X)>0$ for some $X>0$, meaning $\ker(T(X))=\{0\}$, then the same holds for all positive definite matrices by \eqref{eq:ker_T}. 
\end{proof}

Proposition \ref{prop_positivity} shows that a fixed point of full rank guarantees a channel to be strictly positive, whereas the converse does not hold (Example \ref{example_1}).

\begin{remark}
\begin{itemize}
\item[(i)] While strict positivity tells us that one cannot leave the relative interior of all states (the invertible states), the boundary (the non-invertible states) can be either mapped onto the boundary or into the interior. The former is achieved, e.g., by every unitary channel while the latter can be done via a trace projection onto some positive definite state.\smallskip
\item[(ii)] In the generic case the inclusion in \eqref{eq:ker_T} is not an equality. Even worse there is no general statement to make about the rank regarding the image of strictly positive maps (see Example \ref{example_2}).
\end{itemize}
\end{remark}

So far we looked at {\sc sp} and {\sc cp} separately so let us study their interplay via the Kraus representation next. Unsurprisingly the "universal kernel" property from the previous proposition appears in this representation, as well, now in connection with the kernel of the Kraus operators.
\begin{lemma}
Let $T\in \mathbb C^{n\times n}\to \mathbb C^{k\times k}$ be linear and completely positive. Then the following are equivalent.
\begin{itemize}
\item[(i)] $T$ is strictly positive.\smallskip
\item[(ii)] For all sets of Kraus operators $\{K_i\}_{i\in I}$ of $T$ one has $\bigcap_{i\in I}\operatorname{ker}(K_i)=\{0\}$.\smallskip
\item[(iii)] There exist Kraus operators $\{K_i\}_{i\in I}$ of $T$ with $|I|\leq nk$ s.t.~$\bigcap_{i\in I}\operatorname{ker}(K_i)=\{0\}$.\smallskip
\end{itemize}
If one, and hence all, of these conditions hold then $T$ is completely strictly positive. 
\end{lemma}
\begin{proof}
By Proposition \ref{prop_positivity} strict positivity of a positive linear map $T$ is equivalent to $T(\identity)>0$. Observe that, given any set of Kraus operators $\{K_i\}_{i\in I}$ of $T$, one has
$$
\langle \psi,T(\identity)\psi\rangle=\sum_{i\in I}\langle \psi,K_i^\conj K_i\psi\rangle=\sum_{i\in I}\|K_i\psi\|^2
$$
so $\langle\psi,T(\identity)\psi\rangle=0$ holds if and only if $\psi\in\operatorname{ker}(K_i)$ for all $i\in I$. Combining these two things readily implies the above equivalence.

For the additional statement we have to show that $T\otimes\id_m$ is strictly positive for all $m\in\mathbb N$. But again---because $T\otimes\operatorname{id}_m$ is positive by assumption---Proposition \ref{prop_positivity} states that strict positivity is equivalent to $\det((T\otimes\id_m)(\identity))\neq 0$. This holds due to
\begin{align*}
\det((T\otimes\id_m)(\identity_n\otimes\identity_m))=\det(T(\identity))^m\cdot\det(\identity_m)^k=\det(T(\identity))^m\neq 0
\end{align*}
using that $T$ is strictly positive as well as the determinant formula for the Kronecker product \cite[Sec.~4.2]{HJ2}.
\end{proof}

Thus while {\sc cp} and {\sc sp} are incomparable \textit{together} they are equivalent to {\sc csp}.

\begin{lemma}\label{lemma_app_1}
Taking {\sc p, cp} and {\sc sp} as subsets of the set of all linear maps $\mathcal L$ (with the induced subspace topology) the following statements hold. 
\begin{itemize}
\item[(i)] {\sc p} and {\sc cp} are closed, convex subsets of $\mathcal L$.\smallskip
\item[(ii)] {\sc sp} is convex and dense in {\sc p}.\smallskip
\item[(iii)] {\sc sp} is open in the subspace topology induced by {\sc p}.\smallskip
\end{itemize}
\end{lemma}
\begin{proof}
The convexity statements are obvious so we only prove what remains. (i): Let $(T_m)_{m\in\mathbb N}$ be a sequence in {\sc p} which converges to $T\in\mathcal L$. Then for all $A\geq 0$, $\psi\in\mathbb C^n$
$$
\langle \psi,T(A)\psi\rangle=\lim_{m\to\infty} \underbrace{\langle \psi,T_m(A)\psi\rangle}_{\geq 0}\geq 0
$$
so $T(A)\geq 0$ hence $T$ is in {\sc p}. The proof for {\sc cp} is analogous. (ii): Density can be shown constructively: if $T$ is in {\sc p} then $((1-\frac1m)T+\frac1m\id)_{m\in\mathbb N}$ is a sequence in {\sc sp} which approximates $T$.
(iii): We will show, equivalently, that its complement is closed so let $(T_m)_{m\in\mathbb N}$ be a sequence in $\text{\sc p}\setminus\text{\sc sp}$ which converges to some $T\in\text{\sc p}$. Using Proposition \ref{prop_positivity} we get $\det(T(\identity))=\lim_{m\to\infty}\det(T_m(\identity))=0$ by continuity of the determinant so $T\in \text{\sc p}\setminus\text{\sc sp}$ as claimed.
\end{proof}
Be aware that {\sc sp} is not open when taken as a subset of $\mathcal L$, cf.~Example \ref{example_3new}. 
%

\subsection{Characterizations of Strict Positivity}\label{sec:nonstrpos}

The fact that every positive map ``produces'' the same kernel on \textit{all} full-rank states begs the question: what if this kernel is non-zero, that is, what is the footprint of positive maps which are not strictly positive? And how does this kernel manifest? 

It turns out that such maps---up to a unitary channel on the target system---map into a subalgebra of $\mathbb C^{k  \times k }$ the size of which is determined by their action on the identity:
\begin{theorem}\label{mainthm}
Let $T: \mathbb C^{ n \times n }\to \mathbb C^{ k \times k }$ linear and positive be given. Defining \mbox{$m:=\operatorname{dim}\operatorname{ker}(T(\identity))$} the following are equivalent.
\begin{itemize}
\item[(i)] $T$ is not strictly positive.\smallskip
\item[(ii)] There exist pairwise orthonormal vectors $\psi_1,\ldots,\psi_m\in\mathbb C^n$ such that $T(A)\psi_j=0$ as well as $\psi_j^\conj T(A)=0$ for all $A\in\mathbb C^{n\times n}$, $j=1,\ldots,m$ where $m\geq 1$.\smallskip
\item[(iii)] There exists unitary $U\in\mathbb C^{  k\times k }$ such that $\operatorname{im}(\operatorname{Ad}_{U^\conj}\circ T)\subseteq \mathbb C^{ (k-m) \times (k-m) }\oplus 0^{m\times m}$ where $m\geq 1$, that is, for all $A\in\mathbb C^{  n\times  n}$
\begin{equation}\label{eq:T_block_form}
U^\conj T(A)U=\begin{pmatrix} *&0\\0&0_m \end{pmatrix}.
\end{equation}
\item[(iv)] There exists an orthogonal projection $\pi\in\mathbb C^{k\times k}$ of rank $k-m$ where $m\geq 1$ such that \mbox{$\pi T(A)\pi=T(A)$} for all $A\in\mathbb C^{n\times n}$.
\end{itemize}
If in addition $T$ is trace-preserving then $m<k$. 
\end{theorem}
\begin{proof}
(i) $\Rightarrow$ (ii): Define $\mathcal K:=\operatorname{ker}(T(\identity))$ and be aware that $\mathcal K=\bigcap_{A\in \mathbb C^{n  \times n }}{\ker}(T(A))$ by Proposition \ref{prop_positivity}. By assumption $T$ {is} not strictly positive so $T(\identity)$ is not invertible and thus $m=\operatorname{dim}\mathcal K\geq 1$. Now one finds an orthonormal basis $\{\psi_1,\ldots,\psi_m\}$ of $\mathcal K$. Thus $T(A)\psi_j=0$ as well as
$
\psi_j^\conj T(A)=(T(A^\conj )\psi_j)^\conj =0^\conj =0
$
because positivity of $T$ in particular means that $T$ preserves hermiticity.

(ii) $\Rightarrow$ (iii): Define $U$ via $e_{k+1-j}\mapsto \psi_j$ for all $j=1,\ldots,m$ and 
choose $U$ on $e_1,\ldots, e_{k-m}$ such that it becomes unitary. Then for all $j=1,\ldots, m$ and all $A$ one gets
$
U^\conj T(A)Ue_{k+1-j}=U^\conj T(A)\psi_j=0$ as well as $e_{k+1-j}^T U^\conj T(A)U=\psi_j^\conj T(A)U=0
$ which shows \eqref{eq:T_block_form}.

(iii) $\Rightarrow$ (iv): Take $\pi=U(\sum_{i=1}^{k-m}|e_i\rangle\langle e_i|)U^*=\sum_{i=1}^{k-m}|Ue_i\rangle\langle Ue_i|$. 
Finally (iv) $\Rightarrow$ (i) is obvious---this establishes the equivalence of (i) through (iv). If $T$ additionally is trace-preserving then $T(\identity)\neq 0$ so $m<k$.
\end{proof}

In other words lack of strict positivity means that the image of such maps show a ``loss of dimension''. Thus, as a direct application, if the co-domain are the $2\times 2$ matrices then the action of the map is limited to one entry (or it is zero, altogether):

\begin{corollary}\label{coro1}
Let a qudit-to-qubit channel---i.e.~$T: \mathbb C^{ n \times  n}\to\mathbb C^{ 2 \times 2 }$ {\sc cptp}---be given which is not strictly positive. Then there exists $\psi\in\mathbb C^2$, $\langle\psi,\psi\rangle=1$ such that \mbox{$T(A)=\tr(A)|\psi\rangle\langle\psi|$} for all $A$ so $T$ is the trace projection onto a pure state.
\end{corollary}
\begin{proof}
By Theorem \ref{mainthm} there exists $U\in \mathbb C^{  2\times 2 }$ unitary such that 
$$
U^\conj T(A)U=\begin{pmatrix} *&0\\0&0 \end{pmatrix}
$$
for all $A\in\mathbb C^{ n \times  n}$. Because $T$ is trace-preserving $*$ has to be of size $1$ and thus is equal to $\tr(A)$. Hence $T(A)=\tr(A)|Ue_1\rangle\langle Ue_1|$ for all $A$. Choosing $\psi:=Ue_1$ concludes the proof.
\end{proof}
Of course Corollary \ref{coro1} holds not only for channels but all {\sc ptp} maps as the proof did not exploit complete positivity. Moreover the requirement of the final system being a qubit system is essential as one can, unsurprisingly, construct {\sc cptp} maps which are not strictly positive but are not a trace projection (Example \ref{example_5} (i)). 

\begin{remark}\label{rem_ch_heisenberg}
So far we analyzed channels in the Schr{\"o}dinger picture but how does the above phenomenon manifest in the Heisenberg picture where the channels are completely positive and satisfy $T(\identity)=\identity$, hence they are {\sc sp} by definition? Let $T$ be positive but not strictly positive so, using the identity $\tr(T(A)B)=\tr(AT^*(B))$ which relates a linear map and its dual, we get
\begin{align*}
\tr\big(AT^*(B)\big)=\tr\big(T(A)B\big)=\tr\big(\pi T(A)\pi B\big)=\tr\big(T(A) \pi B\pi\big)=\tr\big(AT^*(\pi B\pi)\big)
\end{align*}
with $\pi$ being the corresponding projection from Theorem \ref{mainthm} (iv). This shows that the (pre-)dual channel $T$ of a Heisenberg channel $T^*$ is not strictly positive iff there exists an orthogonal projection such that $T^*(\pi B\pi)=T^*(B)$ for all $B$---in other words $T^*$ in some basis is fully determined by a $(k-m)^2$-dimensional subspace of the input $B$. To substantiate this we refer to Example \ref{example_5} (ii). 
\end{remark}

The final result of this section is motivated by Lemma \ref{lemma_app_1}: the identity is {\sc sp} and the latter forms an open set (relative to {\sc p}) so one finds $\varepsilon>0$ such that every positive linear map $\varepsilon$-close to the identity is strictly positive as well. Now somewhat surprisingly 
using Theorem \ref{mainthm}
this $\varepsilon$ can be lower bounded by one and in the case of trace-preserving maps specified explicitly:

\begin{proposition}\label{lemma_opnorm_dist}
Let $T: (\mathbb C^{ n \times n },\|\cdot\|_1)\to(\mathbb C^{ n \times n },\|\cdot\|_1)$ be linear and positive but not strictly positive. Then
$\|T-\id\|
\geq 1$.
If $T$ additionally is trace-preserving one has $\|T-\id\|=2$ so the distance between $T$ and the identity channel is maximal.
\end{proposition}
\begin{proof}
By Theorem \ref{mainthm} one finds $U\in\mathbb C^{  n\times n }$ unitary such that \mbox{$\operatorname{im}(\operatorname{Ad}_{U^\conj}\circ T)\subseteq \mathbb C^{ (n-m) \times (n-m) }\oplus 0^{m\times m}$} where $m=\operatorname{dim}\operatorname{ker}(T(\identity))\geq 1$. Because $\|\,|Ue_n\rangle\langle Ue_n|\,\|_1=1$ we by unitary equivalence of the trace norm compute
\begin{align}
\|T-\id\|&\geq \big\|T\big( U|e_n\rangle\langle e_n|U^\conj \big)-U|e_n\rangle\langle e_n|U^\conj\big\|_1\notag\\
&=\big\|(\operatorname{Ad}_{U^\conj}\circ T)\big( U|e_n\rangle\langle e_n|U^\conj \big)-|e_n\rangle\langle e_n|\big\|_1\notag\\
&=\|*_{n-m}\,\oplus \,0_{m-1}\oplus (-1)\|_1= \|*_{n-m}\|_1+|-1|\notag\\
&=\big\|T\big(U|e_n\rangle\langle e_n|U^\conj\big)\big\|_1+1=\tr\big(T\big(U|e_n\rangle\langle e_n|U^\conj\big)\big)+1\geq 1\,.\label{eq:laststep}
\end{align}
In the last step we used that $\|A\|_1=\tr(A)$ for all $A\geq 0$ as well as positivity of $T$. If $T$ additionally is trace-preserving then \eqref{eq:laststep} is obviously equal to $2$. To show that this is also an upper bound recall that $\|T\|=\|\id\|=1$ because every {\sc ptp} map is trace norm-contractive (Lemma \ref{lemma_qc_op_bound}) hence by the triangle inequality
$
\|T-\id\|\leq\|T\|+\|\id\|=2
$
which concludes the proof.
\end{proof}
Proposition \ref{lemma_opnorm_dist} gives a necessary criterion for lack of strict positivity which, however, is not sufficient (cf.~Example \ref{example_5c}). 

\begin{remark}[Application to quantum dynamics]\label{rem_markov}
So far we learned that lack of strict positivity comes along with a loss of dimension which, when approaching from a more physical point of view, motivates the following question: given a Markovian \textit{dynamical} system\footnote{
Of course there are a number of classes of quantum-dynamical processes which are not Markovian but as these do not interact as nicely with strict positivity we will sweep them under the rug here. For those more general notions of divisibility we refer the reader to \cite{Wolf08a}.
}, that is, a family $(T_t)_{t\geq 0}$ of {\sc cptp} maps which is continuous in $t$ and satisfies $T_{t+s}=T_tT_s$ for all $s,t\geq 0$ as well as $T_0=\id$, can one determine the exact time when this dimension loss occurs, if at all?

Such processes, also called \textit{quantum-dynamical semigroups}, are necessarily of the form $T_t=e^{tL}$ with the time-independent generator $L:\mathbb C^{n\times n}\to \mathbb C^{n\times n}$ being of specific structure (``{\sc gksl}-form'') comprised of a Hamiltonian and a dissipative part \cite{GKS76,Lindblad76}.
Indeed the semigroup structure guarantees that quantum-dynamical semigroups are strictly positive at all times. 
As for a short proof: given any $t>0$ using continuity of the semigroup in $t$ one finds $m\in\mathbb N$ such that $\|T_{t/m}-T_0\|=\|T_{t/m}-\operatorname{id}\|<2$. This by Proposition \ref{lemma_opnorm_dist} implies that $T_{t/m}$ is a strictly positive channel so---because {\sc sp} forms a semigroup---$T_t=(T_{t/m})^m$ is strictly positive as well.
%
\end{remark}
The result from Remark \ref{rem_markov} also shows that (Markovian) cooling processes---or any relaxation process of semigroup structure the steady state of which is not invertible---have to take infinitely long. This is not too surprising as the dissipation 
happens exponentially in time.
\section{$D$-Matrix Majorization}\label{sec:matrix_d_maj}
Classical majorization $\prec$, that is, $x\prec y$ for $x,y\in\mathbb R^n$, originally is defined via ordering $x,y$ decreasingly and then comparing the partial sums
$
\sum\nolimits_{j=1}^kx_j^\cdot\leq\sum\nolimits_{j=1}^ky_j^\cdot
$
for all $k=1,\ldots,n-1$ as well as $\sum_{j=1}^nx_j=\sum_{j=1}^ny_j$.
Arguably the strongest characterization of majorization reads as follows: setting $\unitvector:=(1,\ldots,1)^T\in\mathbb R^n$ one has $x\prec y$ if and only if there exists a column-stochastic matrix (see footnote \ref{footnote:column_stoch}) $A\in\mathbb R^{n\times n}$ with $A\unitvector=\unitvector$ such that $Ay=x$. Such matrices $A$ are called doubly stochastic and they allow for generalizing majorization to complex vectors as well as
to arbitrary weight vectors $d\in\mathbb R_{++}^n$
where, here and henceforth, $\mathbb R_{++}^n$ denotes the collection of all vectors from $\mathbb R^n$ with strictly positive entries. Indeed the definition of $\prec$ via partial sums cannot extend beyond $\unitvector$: as soon as two entries in $d\in\mathbb R_{++}^n$ differ one loses permutation invariance and reordering the vectors $x,y$ makes a conceptual difference.
For more on 
vector majorization refer to \cite[Ch.~1 \& 2]{MarshallOlkin}. 

Now as stated already in the introduction majorization relative to a strictly positive vector 
$d\in\mathbb R_{++}^n$ 
, denoted by $x\prec_d y$, is defined via the existence of a column-stochastic matrix 
$A\in\mathbb R^{n\times n}$ with $Ad=d$ such that $x=Ay$. Inspired by our previous explanations such matrix $A\in\mathbb R^{n\times n}$
is called $d$-stochastic and setting $d=\unitvector$ recovers $\prec$. A variety of characterizations of $\prec_d$ and $d$-stochastic matrices can be found
in \cite{Joe90,vomEnde19polytope}
: 
\begin{lemma}\label{lemma_char_d_vec}
Let $d\in\mathbb R_{++}^n$ and $x,y\in\mathbb R^n$ be given. The following are equivalent.
\begin{itemize}
\item[(i)] $x\prec_dy$\smallskip
\item[(ii)] $\sum_{j=1}^n d_j \psi(\frac{x_j}{d_j})\leq \sum_{j=1}^n d_j \psi(\frac{y_j}{d_j})$ for all continuous convex functions \mbox{$\psi:D(\psi)\subseteq\mathbb R\to\mathbb R$} such that $\{\frac{x_j}{d_j}\,|\,j=1,\ldots,n\},\{\frac{y_j}{d_j}\,|\,j=1,\ldots,n\}\subseteq D(\psi)$.\smallskip
\item[(iii)] $\sum_{j=1}^n (x_j-td_j)_+\leq\sum_{j=1}^n (y_j-td_j)_+$ for all $t\in\mathbb R$ where $(\cdot)_+:=\max\{\cdot,0\}$.\smallskip
\item[(iv)] $\sum_{j=1}^n (x_j-td_j)_+\leq\sum_{j=1}^n (y_j-td_j)_+$ for all $t\in\{ \frac{x_i}{d_i},\frac{y_i}{d_i}\,|\,i=1,\ldots,n \}$.\smallskip
\item[(v)] $\|x-td\|_1\leq\|y-td\|_1$ 
for all $t\in\mathbb R$ with $\|\cdot\|_1$ being the usual vector $1$-norm.\smallskip
\item[(vi)] $\unitvector^Tx=\unitvector^Ty$ and $\|x-\frac{y_i}{d_i}d\|_1\leq\|y-\frac{y_i}{d_i}d\|_1$ for all $i=1,\ldots,n $.
\end{itemize}
\end{lemma}


The $1$-norm characterization given in (vi) allows to rewrite the $d$-majorization polytope $M_d(y):=\{x\in\mathbb R^n\,|\,x\prec_d y\}$ for any $y\in\mathbb R^n$ as the set of solutions to a nicely structured vector inequality $Mx\leq b$, $M\in\mathbb R^{2^n\times n}$. This 
description of $d$-majorization enables a proof of the existence of an extremal point $z\in M_d(y)$ such that $M_d(y)\subseteq M_e(z)$, that is, there exists some extreme point $z\prec_d y$ which classically majorizes all $x\in M_d(y)$. 
As a recent application $\prec_d\,$ is suitable to analyse and upper bound reachable sets in certain toy models inspired by coupling quantum control systems to a bath of finite temperature \cite{MTNS2020_1}. 

\subsection{Motivation and Definition}\label{subsec:1}
Moving to the matrix case, classical majorization on the level of hermitian matrices uses their ``eigenvalue vector'' $\vec\lambda(\cdot)$ arranged in any order with multiplicities counted.
More precisely for $A,B\in \mathbb C^{n\times n}$ hermitian, $A$ is said to be majorized by $B$ if \mbox{$\vec\lambda(A)\prec\vec\lambda(B)$}, cf. \cite{Ando89}. The most na{\"i}ve approach to define \mbox{$D$-majorization} on matrices would be to replace $\prec$ by $\prec_d$ and leave the rest as it is. However just as in the vector case such a definition depends on the eigenvalues' arrangement in $\vec\lambda$ which is infeasible due to the lack of permutation invariance of $d$---unless, of course, $d=\unitvector$. 
The most natural way out of this dilemma is to characterize classical majorization on matrices via quantum channels which have the identity matrix 
as a fixed point \cite[Thm.~7.1]{Ando89}:
\begin{lemma}\label{thm_ando_7_1}
Let $A,B\in\mathbb C^{n\times n}$ be hermitian. The following are equivalent.
\begin{itemize}
\item[(i)] $A\prec B$, that is, $\vec\lambda(A)\prec\vec\lambda(B)$.\smallskip
\item[(ii)] There exists $ T\in Q(n)$ such that $T(B)=A$, $ T(\identity_n)=\identity_n$.\smallskip
\item[(iii)] There exists $ T:\mathbb C^{n\times n}\to \mathbb C^{n\times n}$ linear and {\sc ptp} such that $T(B)=A$, $ T(\identity_n)=\identity_n$.
\end{itemize}
\end{lemma}
Therefore it seems utmost reasonable to generalize majorization on square matrices as follows.
\begin{definition}\label{defi_matrix_D_maj}
Let $D\in\mathbb C^{n\times n}$ positive definite and $A,B\in\mathbb C^{n\times n}$. Now $A$ is 
said to be \textit{$D$-majorized} by $B$, denoted by $A\prec_D B$, if there exists $T\in Q(n)$ such that $T(B)=A$ 
and $T(D)=D$.
\end{definition}


Now Lemma \ref{lemma_convex_subsemigroup} implies that $\prec_D$ admits a convex structure.
\begin{lemma}\label{lemma_conv_d_maj_matrices}
Let $A,B,C,D\in\mathbb C^{n\times n}$ with $D>0$ be given. If $A\prec_D C$ and $B\prec_D C$, then $\lambda A+(1-\lambda)B\prec_D C$ for all $\lambda\in [0,1]$.
\end{lemma}
\begin{remark}[Semigroup majorization]\label{rem_semigroup}
Comparing to \cite[Def.~4.6]{Parker96} $\prec_D$ coincides with the semigroup majorization induced by $Q_D(n)$, a concept also treated in \cite[Ch.~14.C]{MarshallOlkin}. Moreover as $Q_D(n)$ is convex and compact (Lemma \ref{lemma_convex_subsemigroup}) one can consider the set $Q^E_D(n)$ of extreme points of $Q_D(n)$ which can be abstractly characterized using, e.g., \cite[Thm.~5]{Choi75} or \cite[Thm.~4]{Verstraete}. Now for any $A,B\in\mathbb C^{n\times n}$, $A\prec_D B$ holds if and only if $A$ lies in the convex hull of the set $\lbrace T(B)\,|\,T\in Q^E_D(n)\rbrace$, cf. \cite[Ch.~14, Obs.C2.(iii)]{MarshallOlkin}.
\end{remark}
Be aware that there is a global unitary degree of freedom here: given square matrices $A,B,D$ with $D>0$ and any unitary transformation $U$ one has $A\prec_D B$ if and only if $UAU^*\prec_{UDU^*} UBU^*$, with the proof being a straightforward computation. This we use to require w.l.o.g.~that $D=\diag(d)$ for some strictly positive vector $d\in\mathbb R_{++}^n$. Now the relation between $\prec_D$ and $\prec_d$ 
reads as follows.
\begin{corollary}\label{coro_d_maj_diag}
Let $d\in\mathbb R_{++}^n$ and matrices $A,B\in\mathbb C^{n\times n}$ be given. Then the following statements hold.
\begin{itemize}
\item[(i)] If $A$ is diagonal and $(\langle e_j,Ae_j\rangle)_{i=1}^n\prec_d (\langle e_j,Be_j\rangle)_{i=1}^n$, then $A\prec_D B$.\smallskip
\item[(ii)] If $B$ is diagonal and $A\prec_D B$, then $(\langle e_j,Ae_j\rangle)_{i=1}^n\prec_d (\langle e_j,Be_j\rangle)_{i=1}^n$.
\end{itemize}
Additionally, if $A,B$ both are diagonal then $A\prec_D B$ holds if and only if the diagonal of $B$ $d$-majorizes that of $A$, that is, $(\langle e_j,Ae_j\rangle)_{j=1}^n\prec_d (\langle e_j,Be_j\rangle)_{j=1}^n$.
\end{corollary}
\begin{proof}
We only prove (i) and (ii) because the additional statement directly follows from those. For convenience let $a_{j}:=\langle e_j,Ae_j\rangle$, $b_{j}:=\langle e_j,Be_j\rangle$ for all $j=1,\ldots,n$ as well as $a:=(a_{j})_{j=1}^n$, $b:=(b_{j})_{j=1}^n\in\mathbb C^n$.

(i): By assumption there exists $d$-stochastic $M\in\mathbb R^{n\times n}$ with $a=Mb$. Define $T$ via
\begin{align}
T:\mathbb C^{n\times n}&\to \mathbb C^{n\times n}\notag\\
|e_i\rangle\langle e_j|&\mapsto\begin{cases}
0&\text{ if }i\neq j\\
\sum\nolimits_{k=1}^n M_{ki} |e_k\rangle\langle e_k|&\text{ if }i= j
\end{cases}\label{eq:T}
\end{align}
and its linear extension onto all of $\mathbb C^{n\times n}$. Thus for any $X\in \mathbb C^{n\times n}$
\begin{align*}
T(X)=\sum_{i,j=1}^n X_{ij}T(|e_i\rangle\langle e_j|)=\sum_{i,k=1}^n X_{ii}M_{ki} |e_k\rangle\langle e_k|=\sum_{k=1}^n \Big(\sum_{i=1}^n X_{ii}M_{ki}  \Big)  |e_k\rangle\langle e_k|
\end{align*}
where $X_{ij}=\langle e_i,Xe_j\rangle$. Thus $M$ being column-stochastic implies that $T$ is trace-preserving.
Moreover the Choi matrix of $T$ is diagonal with non-negative entries by \eqref{eq:T} so $C(T)\geq 0$ which shows $T$ is completely positive by Lemma \ref{lemma_choi_matrix}.
All that is left to check now is $T(B)=A$ and $T(D)=D$. Indeed
\begin{align*}
T(B)=\sum\nolimits_{j=1}^n b_j T( |e_j\rangle\langle e_j|)&=\sum\nolimits_{i=1}^n\Big(\sum\nolimits_{j=1}^nM_{ij}b_j\Big) |e_i\rangle\langle e_i|\\
&=\sum\nolimits_{i=1}^n(Mb)_i |e_i\rangle\langle e_i|=\sum\nolimits_{i=1}^na_i |e_i\rangle\langle e_i|=A
\end{align*}
and
\begin{align*}
T(D)=\sum_{j=1}^n d_j T( |e_j\rangle\langle e_j|)=\sum_{i=1}^n\Big(\sum_{j=1}^nM_{ij}d_j\Big) |e_i\rangle\langle e_i|=\sum_{i=1}^n\underbrace{(Md)_i}_{=d_i} |e_i\rangle\langle e_i|=D\,.
\end{align*}

(ii): Given $T\in Q_D(n)$ with $A=T(B)$ define
$
M=\big(\langle e_i,T( |e_j\rangle\langle e_j|)e_i\rangle\big)_{i,j=1}^n\in\mathbb R^{n\times n}\,.
$
As above one verifies that $M$ is $d$-stochastic 
and
\begin{equation}
Mb=\sum_{i=1}^n (Mb)_i \,e_i=\sum_{i,j=1}^n \langle e_i,T( |e_j\rangle\langle e_j|)e_i\rangle b_j \,e_i = \sum_{i=1}^n \langle e_i,T(B)e_i\rangle\,e_i =a\,.\tag*{\qedhere}
\end{equation}
\end{proof}

\subsection{Characterizing $\prec_D$ for Hermitian Matrices}\label{subsec:1b}
As explained in the previous chapter, for the rest of this article we w.l.o.g.~require the positive definite matrix $D\in\mathbb R^{n\times n}$  to be diagonal, i.e.~$D=\diag(d)$ for some $d\in\mathbb R_{++}^n$.
\begin{proposition}\label{thm_char}
Let $d\in\mathbb R_{++}^2$ and $A,B\in \mathbb C^{2\times 2}$ hermitian be given. Then the following statements are equivalent.
\begin{itemize}
\item[(i)] $A\prec_D B$, that is, there exists $T\in Q_D(2)$ such that $T(B)=A$.\smallskip
\item[(ii)] There exists $T:\mathbb C^{2\times 2}\to \mathbb C^{2\times 2}$ linear and {\sc ptp} such that $T(B)=A$.\smallskip
\item[(iii)] $\|A-tD\|_1\leq \|B-tD\|_1$ for all $t\in\mathbb R$.\smallskip
\item[(iv)] $\tr(A)=\tr(B)$ and $\|A-b_iD\|_1\leq \|B-b_iD\|_1$ for $i=1,2$ as well as for the generalized fidelity
$$
\big\|\sqrt{A-b_1D}\sqrt{b_2D-A}\big\|_1\geq \big\|\sqrt{B-b_1D}\sqrt{b_2D-B}\big\|_1
$$
where $\sigma(D^{-1/2}BD^{-1/2})=\{b_1,b_2\}$ ($b_1\leq b_2$) with, here and henceforth, $\sigma(\cdot)$ being the spectrum
.
\end{itemize}
\end{proposition}
\begin{proof}
``(i) $\Rightarrow$ (ii)'': Obvious. ``(ii) $\Rightarrow$ (iii)'': By assumption $T(B-tD)=A-tD$ for all $t\in\mathbb R$ so the claim follows from Lemma \ref{lemma_qc_op_bound}.

``(iii) $\Rightarrow$ (i)'': Define 
\begin{align*}
a_1&:=\min\sigma(D^{-1/2}AD^{-1/2})\qquad b_1:=\min\sigma(D^{-1/2}BD^{-1/2})\\
a_2&:=\max\sigma(D^{-1/2}AD^{-1/2})\qquad b_2:=\max\sigma(D^{-1/2}BD^{-1/2})\,.
\end{align*}
Then\footnote{
The key here is the following well-known result: let $X\in\mathbb C^{n\times n}$ be hermitian with smallest eigenvalue $x_m$ and largest eigenvalue $x_M$. Then
$
-X+t\identity_n\geq 0$ if and only if $t\geq x_M$ and $X-t\identity_n\geq 0$ if and only if $t\leq x_m\,.
$
This is evident due to $x_m\|y\|^2\leq \langle y,Xy\rangle\leq x_M\|y\|^2$ for all $y\in\mathbb C^n$ (cf.~\cite[Theorem 4.2.2]{HJ1}).\label{footnote_herm_pos}
}
using that $Y\mapsto D^{1/2}YD^{1/2}$ is positive with positive inverse
\begin{align*}
\begin{split}
\hphantom{-}A-tD\geq 0, \hphantom{-}B-tD\geq 0&\qquad\text{ for all }t\leq s:=\min\{a_1,b_1\}\\
-A+tD\geq 0,-B+tD\geq 0&\qquad\text{ for all }t\geq r:=\max\{a_2,b_2\}\,.
\end{split}
\end{align*}
Because the trace norm of a hermitian matrix is equal to its trace if and only if it is positive this implies
$$
\tr(A)-s\tr(D)=\|A-sD\|_1\leq\|B-sD\|_1=\tr(B)-s\tr(D)\ \Rightarrow\ \tr(A)\leq\tr(B)
$$
and $\|A-rD\|_1\leq\|B-rD\|_1$ shows $-\tr(A)\leq-\tr(B)$ so combined $\tr(A)=\tr(B)$. Thus for arbitrary $t_0<s $ we may define
\begin{equation}\label{eq:tilde_A_B}
\tilde A:=\frac{A-t_0D}{\tr(A)-t_0\tr(D)}\qquad \tilde B:=\frac{B-t_0D}{\tr(A)-t_0\tr(D)}
\end{equation}
so by our previous considerations $\tilde A,\tilde B>0$, $\tilde A,\tilde B\in\mathbb D(\mathbb C^2)$ and $\|\tilde A-t'D\|_1\leq \|\tilde B-t'D\|_1$ for all $t'\in\mathbb R$ by direct computation. Now the Alberti-Uhlmann theorem \cite{AlbertiUhlmann80} guarantees the existence of a {\sc cptp} map $T$ such that $T(\tilde B)=\tilde A$ and $T(D)=D$. Because $T$ is linear and $\tr(A)=\tr(B)$ one even has $T(B)=A$ which shows $A\prec_DB$.\smallskip

``(i) $\Leftrightarrow$ (iv)'': Again we want to reduce this problem from hermitian matrices to states to make use of \cite[Theorem 6]{HeinosaariWolf12}. To see this---due to $\tr(A)=\tr(B)$---as before one finds $t_0<s$ such that $A-t_0D,B-t_0D>0$ so define $\tilde A,\tilde B\in\mathbb D(\mathbb C^2)$ as in \eqref{eq:tilde_A_B}. Using that for all $X\in\mathbb C^{n\times n}$ hermitian and all $c_1,c_2\in\mathbb R$ one has $\sigma(c_1X+c_2\identity_n)=c_1\sigma(X)+c_2$
$$
\sigma(D^{-1/2}\tilde BD^{-1/2})=\frac{\sigma(D^{-1/2}BD^{-1/2})-t_0}{\tr(A)-t_0\tr(D)}\,.
$$
With this it is easy to see that $\|A-tD\|_1\leq\|B-tD\|_1$ for all $t\in\sigma(D^{-1/2}BD^{-1/2})$ is equivalent to $\|\tilde A-t'D\|_1\leq\|\tilde B-t'D\|_1$ for all $t'\in\sigma(D^{-1/2}\tilde BD^{-1/2})=\{\tilde b_1,\tilde b_2\}$, $0<\tilde b_1\leq\tilde b_2$. By the same argument as in ``(iii) $\Rightarrow$ (i)''
\begin{align*}
\begin{split}
\tilde B-t'D\geq 0\quad&\text{ if and only if }\quad t'\leq \min\sigma(D^{-1/2}\tilde BD^{-1/2})=\tilde b_1\\
\tilde B-t'D\leq 0\quad&\text{ if and only if }\quad t'\geq \max\sigma(D^{-1/2}\tilde BD^{-1/2})=\tilde b_2
\end{split}
\end{align*}
which shows $\inf(\tilde B/D):=\sup\{t'\in\mathbb R\,|\,\tilde B-tD\geq 0\}=\tilde b_1$ and
\begin{align*}
\inf (D/\tilde B)&=\sup\{t'\in\mathbb R\,|\,D-t'\tilde B\geq 0\}\overset{D>0}=\sup\{t'>0\,|\,D-t'\tilde B\geq 0\}\\
&=\sup\{t'>0\,|\,B-\tfrac{1}{t'}D\leq 0\}=(\tilde b_2)^{-1}\,.
\end{align*}
With this the following statements are equivalent:
\begin{itemize}
\item $A\prec_D B$\smallskip
\item $\tilde A\prec_D \tilde B$ (linearity)\smallskip
\item $\tilde A-\tilde b_1D\geq 0$, $\tilde A-\tilde b_2D\leq 0$ as well as the trace norm inequality $\|\sqrt{\vphantom{\frac11}\tilde A-\tilde b_1D}\sqrt{\vphantom{\frac11}\tilde b_2D-\tilde A}\|_1\geq \|\sqrt{\vphantom{\frac11}\tilde B-\tilde b_1D}\sqrt{\vphantom{\frac11}\tilde b_2D-\tilde B}\|_1$ (due to \cite[Theorem 6]{HeinosaariWolf12} \& pulling out positive constants).\smallskip
\item $\|\tilde A-t'D\|_1\leq\|\tilde B-t'D\|_1$ for all $t'\in\sigma(D^{-1/2}\tilde BD^{-1/2})$  as well as $\|\sqrt{\vphantom{\frac11}\tilde A-\tilde b_1D}\sqrt{\vphantom{\frac11} \tilde A-\tilde b_2D}\|_1\geq \|\sqrt{\vphantom{\frac11}\tilde B-\tilde b_1D}\sqrt{\vphantom{\frac11}\tilde B-\tilde b_2D}\|_1$
\end{itemize}
For the latter note that
$
\|\tilde B-\tilde b_1D\|_1=\tr(\tilde B-\tilde b_1D)=\tr(\tilde A-\tilde b_1D)
$
and similarly $\|\tilde B-\tilde b_2D\|_1=-\tr(\tilde A-\tilde b_2D)$---therefore the trace norm conditions are equivalent to the positivity conditions $\tilde A-\tilde b_1D,-\tilde A+\tilde b_2D\geq 0$ because a hermitian matrix $X\in\mathbb C^{n\times n}$ is positive semi-definite if and only if $\|X\|_1=\tr(X)$ if and only if $\|X\|_1\leq\tr(X)$. Now by construction the last point from the above list is in turn equivalent to (iv) as $\tilde A-\tilde b_iD,\tilde B-\tilde b_iD$ equal $A-b_iD$, $B-b_iD$ for $i=1,2$ up to global positive constant.
\end{proof}

It may be possible to prove Proposition \ref{thm_char} (iv) $\Rightarrow$ (i) by applying Proposition \ref{prop_exist_channel} to $A-b_1D, B-b_1D$ and $A-b_2D,B-b_2D$, respectively, to get two trace-preserving maps $T_1,T_2$ mapping $B$ to $A$ and having $D$ as fixed point---because $B-tD$ is rank-deficient if and only if $t\in\sigma(D^{-1/2}BD^{-1/2})$---and the fidelity condition might ensure that one of these two is completely positive. However even if this works then one would, most likely, end up with an argument rather close to \cite{AlbertiUhlmann80} so we save ourselves the bother. 

\begin{remark}
\begin{itemize}
\item[(i)] The characterizations from Proposition \ref{thm_char} do not generalize to dimensions larger than $2$. To see this Heinosaari et al.~\cite{HeinosaariWolf12} gave a counterexample to the Alberti-Uhlmann theorem in higher dimensions which pertains to our case. Consider the hermitian matrices
\begin{equation}\label{eq:counterex_heinosaari}
A=\begin{pmatrix} 2&1&0\\1&2&-i\\0&i&2 \end{pmatrix}\quad B=\begin{pmatrix} 2&1&0\\1&2&i\\0&-i&2 \end{pmatrix}\quad D=\begin{pmatrix} 2&1&0\\1&2&1\\0&1&2 \end{pmatrix}\,.
\end{equation}
Indeed $\sigma(D)=\{2,2+\sqrt{2},2-\sqrt{2}\}$ so $D>0$. Obviously $B^T=A$ and $D^T=D$ so because the transposition map is well-known to be linear, positive and trace-preserving one has
$
\|A-tD\|_1=\|(B-tD)^T\|_1\leq \|B-tD\|_1
$
for all $t\in\mathbb R$ by Lemma \ref{lemma_qc_op_bound}. But there exists no {\sc cptp} map, i.e.~no $T\in Q(n)$ such that $T(B)=A$ and $T(D)=D$ as shown in \cite[Proposition 6]{HeinosaariWolf12}.\smallskip
\item[(ii)] Usually a characterization of any generalized form of majorization via convex functions is much sought-after. A reasonable extension of Proposition \ref{thm_char}
\textit{would} be that $A\prec_D B$ if and only if
\begin{equation}\label{eq:matrix_convex_cond}
\tr\big(D\psi(D^{-1/2}AD^{-1/2})\big)\leq \tr\big(D\psi(D^{-1/2}BD^{-1/2})\big)
\end{equation}
for all $\psi:\mathbb R\to\mathbb R$ matrix convex\footnote{
A matrix convex function (cf.~\cite{Kraus36,Bendat55,Ando79,Bhatia}) is a map $\psi:\mathbb R\to\mathbb R$ which acts on hermitian matrices via the spectral theorem and then satisfies
$$
\psi(\lambda A+(1-\lambda)B)\leq \lambda \psi(A)+(1-\lambda) \psi(B)\quad\text{ for all }\lambda\in [0,1], A,B\in\mathbb C^{n\times n}\text{ hermitian}
$$
and all $n\in\mathbb N$ where $\leq$ is the partial ordering on the hermitian matrices induced via positive semi-definiteness.
}
---this is justified by the fact that if $A,B,D$ are all diagonal then \eqref{eq:matrix_convex_cond} reduces to the convex function-condition from the vector case (Lemma \ref{lemma_char_d_vec} (ii)). One can even show that \eqref{eq:matrix_convex_cond} is necessary for some $T\in Q_D(n)$ to satisfy $T(B)=A$ \cite[Thm.~2.1]{Li96}. 
However, condition \eqref{eq:matrix_convex_cond} is also disproven by the matrices in \eqref{eq:counterex_heinosaari} in the same way as above for the following reason: because $D=D^T$ one has
$$
\sigma(D^{-1/2}AD^{-1/2})=\sigma((D^{-1/2}AD^{-1/2})^T)=\sigma(D^{-1/2}A^TD^{-1/2})
$$
so no matrix convex $\psi$---as those act via functional calculus---can distinguish $D^{-1/2}AD^{-1/2}$ from $D^{-1/2}A^TD^{-1/2}$ which readily implies equality in \eqref{eq:matrix_convex_cond}.
\end{itemize}
\end{remark}
While there exist general conditions for the existence of a quantum channel which maps a finite input set of states to an output set of same cardinality \cite{Huang12} these are rather technical and not really applicable in practice. For now characterizing $\prec_D$ beyond two dimensions (via some easy-to-verify inequalities) remains an open problem.

\subsection{The Role of Strict Positivity in Order, Geometrical, and Other Properties of $\prec_D$}\label{subsec:2}
There are two results we will present here for which our analysis of strict positivity in Chapter \ref{sec:strictpos} was essential. The first of these is almost immediate:
\begin{corollary}\label{thm_full_rank}
Let $d\in\mathbb R_{++}^n$ as well as $\rho,\omega\in\mathbb D(\mathbb C^{n})$ be given. If $\rho$ is of full rank and $\omega\prec_D\rho$ then $\omega$ is of full rank, as well.
\end{corollary}
\begin{proof}
By assumption there exists $T\in Q_D(n)$ such that $T(\rho)=\omega\in\mathbb D(\mathbb C^{n})$. Now \mbox{$T(D)=D>0$} by Proposition \ref{prop_positivity} implies that $\rho>0$ is mapped to something positive definite (hence of full rank) again. 
\end{proof}
For the second connection we have to dive into order properties of $D$-majorization first. Some simple observations: just like in the vector case $\prec_D$ is a preorder but it is not a partial order. To see the latter---even if the eigenvalues of $D$ differ pairwise---consider the counterexample for $\prec_d$ given in \cite[Remark 2 (iv)]{vomEnde19polytope} which transfers onto $\prec_D$ via Corollary \ref{coro_d_maj_diag}. Next let us investigate minimal and maximal elements of $\prec_D$ for which we first need the following lemma.
\begin{lemma}\label{thm_channel_pure_max_state}
Let $d\in\mathbb R_{++}^n$, $\rho\in\mathbb D(\mathbb C^{n})$ and $j\in\lbrace 1,\ldots,n\rbrace$ be given. Then $\rho\prec_D |e_j\rangle\langle e_j|$ if and only if $D-d_j\rho\geq 0$.
\end{lemma}
\begin{proof}
``$\Rightarrow$'' : By definition there exists $T\in Q_D(n)$ such that $T( |e_j\rangle\langle e_j|)=\rho$. Note that $D-d_j |e_j\rangle\langle e_j|\geq 0$ as the l.h.s.~is a diagonal matrix with non-negative entries, so linearity and positivity of $T$ imply
\begin{align*}
0\leq T\big(D-d_j |e_j\rangle\langle e_j|\big)=T(D)-d_j T\big( |e_j\rangle\langle e_j|\big)=D-d_j\rho\,.
\end{align*}
\noindent ``$\Leftarrow$'' : The case $n=1$ is trivial so assume $n> 1$. As $D-d_j\rho\geq 0$ by assumption, 
$
\omega:=\frac{D-d_j\rho}{\unitvector^Td-d_j}\in\mathbb D(\mathbb C^{n})
$ and $d_j\rho+(\unitvector^Td-d_j)\omega=D$. With this, define a linear map $T:\mathbb C^{n\times n}\to\mathbb C^{n\times n}$ via $T(| e_i\rangle\langle e_k|)=0$ whenever $i\neq k$ and
\begin{align*}
T( |e_i\rangle\langle e_i|)=\begin{cases} \omega &i\neq j\\ \rho &i=j \end{cases}
\end{align*}
as well as its linear extension to all of $\mathbb C^{n\times n}$. Now $T$ is trace-preserving and
\begin{align*}
T(D)=d_jT( |e_j\rangle\langle e_j|)+\sum\nolimits_{i=1,i\neq j}^n d_i T( |e_i\rangle\langle e_i|)=d_j\rho+(\unitvector^Td-d_j)\omega=D\,,
\end{align*}
as well as $T( |e_j\rangle\langle e_j|)=\rho$. For complete positivity, consider the Choi matrix
\begin{align*}
C(T)=\begin{pmatrix} T( |e_1\rangle\langle e_1|)&&0\\&\ddots&\\0&&T( |e_n\rangle\langle e_n|) \end{pmatrix}=\underbrace{\omega\oplus\ldots\oplus \omega}_{j-1\text{ times}}\oplus \rho \oplus\underbrace{\omega\oplus\ldots\oplus \omega}_{n-j\text{ times}}
\end{align*}
which is a block-diagonal matrix built from states so $C(T)\geq 0$ and thus $T$ is completely positive by Lemma \ref{lemma_choi_matrix}. Hence we constructed $T\in Q_D(n)$ with $T( |e_j\rangle\langle e_j|)=\rho$ which concludes the proof.
\end{proof}
\begin{remark}
For $d=\unitvector$ one has $D-d_j\rho=\identity_n-\rho\geq 0$ for all $\rho\in\mathbb D(\mathbb C^{n})$ so we recover the well-known result that every pure state is maximal in $\mathbb D(\mathbb C^{n})$ w.r.t.~$\prec\,$. 
This implies that if the rank of some $\rho\in\mathbb D(\mathbb C^{n})$ is larger than one then there exists no $\psi\in\mathbb C^n$ such that $ |\psi\rangle\langle\psi|\prec\rho$. For general $D$-majorization this fails---consider again the example from \cite[Remark 2 (iv)]{vomEnde19polytope} together with Corollary \ref{coro_d_maj_diag}.
However there still is the weaker result that full rank is preserved under $\prec_D$ (Corollary \ref{thm_full_rank}).
\end{remark}

With these tools at hand we, like in the vector case, show there exists a minimal and maximal state with respect to $\prec_D$ and we can even characterize uniqueness: 
\begin{theorem}\label{d_matrix_minmax}
Let $d\in\mathbb R_{++}^n$ be given and let
\begin{align*}
\mathfrak h_d&:=\lbrace X\in\mathbb C^{n\times n}\,|\,\tr(X)=\unitvector^Td\rbrace\\
\mathfrak h_d^+&:=\lbrace X\in\mathbb C^{n\times n}\,|\,X\geq 0\text{ and }\tr(X)=\unitvector^Td\rbrace
\end{align*}
be the trace hyperplane induced by $d$ within the complex and the positive semi-definite matrices, respectively. The following statements hold.
\begin{itemize}
\item[(i)] $D$ is the unique minimal element in $\mathfrak h_d$ with respect to $\prec_D$.\smallskip
\item[(ii)] $(\unitvector^Td)|e_k\rangle\langle e_k|$ is maximal in $\mathfrak h_d^+$ with respect to $\prec_D$ where $k$ is chosen such that $d_k$ is minimal in $d$. It is the unique maximal element in $\mathfrak h_d^+$ with respect to $\prec_D$ if and only if $d_k$ is the unique minimal element of $d$.
\end{itemize}
\end{theorem}
\begin{proof}
(i): To see $D\prec_D A$ for arbitrary $A\in\mathfrak h_d$, consider $T:\mathbb C^{n\times n}\to \mathbb C^{n\times n}$, \mbox{$X\mapsto D \tr(X)$} which is in $Q(n)$ \cite[Ex.~5.3]{Hayashi06} and satisfies $T(D)=D=T(A)$. Uniqueness is evident as $D$ has to be a fixed point of $T$.

(ii): W.l.o.g.~$\unitvector^Td=1$ so $\mathfrak h_d^+=\mathbb D(\mathbb C^{n})$
. Let arbitrary $\rho\in\mathbb D(\mathbb C^{n})$ be given. Because $d_k$ is the minimal eigenvalue of $D$ one finds $D-d_k\rho\geq 0$ due to
\begin{equation}\label{eq:D_pos_matrix}
\begin{split}
\langle x,(D-d_k\rho)x\rangle=\langle x,Dx\rangle-d_k\langle x,\rho x\rangle&\geq d_k\|x\|^2-d_k\|x\|^2\|\rho\|\\
&\geq d_k\|x\|^2(1-\|\rho\|_1)=0
\end{split}
\end{equation}
which holds for all $x\in\mathbb C^n$. Here we used $\|\rho\|_1=\tr(\rho)=1$ as $\rho\geq 0$. Now by Lemma \ref{thm_channel_pure_max_state} this implies $\rho\prec_D  |e_k\rangle\langle e_k|$.

To prove uniqueness first assume that $d_k$ is the unique minimal element of $d$ and that $\omega\in\mathbb D(\mathbb C^{n})$ is also maximal w.r.t. $\prec_D$. Thus $ |e_k\rangle\langle e_k|\prec_D \omega$, that is, there exists $T\in Q_D(n)$ such that $T(\omega)= |e_k\rangle\langle e_k|$. We can diagonalize \mbox{$\omega=\sum_{i=1}^r w_i|g_i\rangle\langle g_i|$} with $w_1,\ldots, w_r> 0$, $\sum_{i=1}^r w_i=1$ and some orthonormal system $(g_i)_{i=1}^r$ in $\mathbb C^n$ where $r\in\lbrace 1,\ldots,n\rbrace$. Then
$$
 |e_k\rangle\langle e_k|=T(\omega)=\sum\nolimits_{i=1}^r w_i T(|g_i\rangle\langle g_i|)
$$
meaning we expressed a pure state as a convex combination of density matrices. But by Lemma \ref{lemma_pure_extreme} this forces $T(|g_i\rangle\langle g_i|)= |e_k\rangle\langle e_k|$ for all $i=1,\ldots,r$.
Now $D-d_k|g_i\rangle\langle g_i|\geq 0$ for all $i$ by \eqref{eq:D_pos_matrix}---actually this matrix is positive definite if and only if $g_i$ and $e_k$ are linearly independent if and only if\,\footnote{
While these equivalences are straightforward to check the main ingredients are the estimate 
$$
\langle x,Dx\rangle=\sum\nolimits_{i=1}^n d_i|\langle e_i,x\rangle|^2\geq d_k\sum\nolimits_{i=1}^n|\langle e_i,x\rangle|^2=d_k\|x\|^2
$$
for all $x\in\mathbb C^n$---with equality if and only if $x=\lambda e_k$ for some $\lambda\in\mathbb C$ because $d_k$ is the unique minimal entry of $d$---as well as the renowned fact that equality in the Cauchy-Schwarz inequality holds if and only if one vector is a multiple of the other.
}
$|g_i\rangle\langle g_i|\neq |e_k\rangle\langle e_k|$. However, $D-d_k|g_i\rangle\langle g_i|>0$ would imply $T(D-d_k|g_i\rangle\langle g_i|)>0$ by Proposition \ref{prop_positivity}---due to $T(D)=D>0$---so
\begin{align*}
0<\langle e_k,T(D-d_k|g_i\rangle\langle g_i|)e_k\rangle&= \langle e_k,T(D)e_k\rangle-d_k\langle e_k,T(|g_i\rangle\langle g_i|)e_k\rangle\\
&= \langle e_k,De_k\rangle-d_k|\langle e_k,e_k\rangle|^2=0
\end{align*}
for all $i=1,\ldots,r$, an obvious contradiction. Hence $|g_i\rangle\langle g_i|= |e_k\rangle\langle e_k|=\omega$.

Finally, assume there exist $k,k'\in\lbrace1,\ldots,n\rbrace$ with $k\neq k'$ such that $d_k=d_{k'}$ is minimal in $d$. Then $ |e_k\rangle\langle e_k|$ and $|e_{k'}\rangle\langle e_{k'}|$ are both maximal with respect to $\prec_D$ by the same argument as above, hence no uniqueness. This concludes the proof.
\end{proof}
Note that strict positivity of every $T\in Q_D(n)$ was the key in proving uniqueness of the maximal element of $\prec_D$, assuming the corresponding eigenvalue of $D$ is simple.
\begin{remark}
From a physical point of view this is precisely what one expects: from the state with the largest energy one can generate every other state (in an equilibrium-preserving manner) and there is no other state with this property.
\end{remark}

As described at the start of Chapter \ref{sec:matrix_d_maj}, just like in the vector case, considering the set of all matrices which are $D$-majorized by some $X\in\mathbb C^{n\times n}$ is of interest for analyzing reachable sets of certain quantum control problems. Therefore define
\begin{equation}\label{eq:def_MD}
\begin{split}
M_D:\mathcal P(\mathbb C^{n\times n})&\to\mathcal P(\mathbb C^{n\times n})\\
S&\mapsto \bigcup\nolimits_{Y\in S}\lbrace X\in\mathbb C^{n\times n}\,|\,X\prec_D Y\rbrace
\end{split}
\end{equation}
where $\mathcal P$ denotes the power set. For convenience $M_D(X):=M_D(\{X\})$ for any \mbox{$X\in\mathbb C^{n\times n}$}
. Then Lemma \ref{lemma_qc_op_bound} as well as Remark \ref{rem_semigroup} lead to the following.
\begin{theorem}\label{lemma_R_d_closed}
Let $d\in\mathbb R_{++}^n$, $B\in\mathbb C^{n\times n}$ and a subset $P\subseteq \mathbb C^{n\times n}$ be given. The following statements hold.
\begin{itemize}
\item[(i)] $M_D(A)$ is convex for all $A\in\mathbb C^{n\times n}$.\smallskip
\item[(ii)] $M_D$ as an operator on $\mathcal P(\mathbb C^{n\times n})$ 
is a closure operator\,\footnote{An operator $J$ on the power set $\mathcal P(S)$ of a set $S$ is called \textit{closure operator} or \textit{hull operator} if it is extensive ($X\subseteq J(X)$), increasing ($X\subseteq Y\,\Rightarrow\,J(X)\subseteq J(Y)$) and idempotent ($J(J(X))=J(X)$) for all $X,Y\in\mathcal P(S)$, cf., e.g., \cite[p.~42]{Cohn81}.}.\smallskip
\item[(iii)] If $P$ is compact, then $ M_D(P)$ is compact.\smallskip
\item[(iv)] If $A$ is an extreme point of $M_D(B)$ then there exists an extreme point $T$ of $Q_D(n)$ such that $T(B)=A$.
\end{itemize}
\end{theorem}
\begin{proof}
(i): Simple consequence of Lemma \ref{lemma_conv_d_maj_matrices}. (ii): Obviously, $M_D$ is extensive and increasing. For idempotence ($M_D\circ M_D=M_D$), ``$\,\subseteq\,$'' follows from $Q_D(n)$ forming a semigroup and ``$\,\supseteq\,$'' is due to $\operatorname{id}_n\in Q_D(n)$. (iii): As all norms in finite dimensions are equivalent, and thus all induced topologies, we can w.l.o.g.~equip $\mathbb C^{n\times n}$ with the trace norm. As $P$ by assumption is bounded Lemma \ref{lemma_qc_op_bound} implies that $M_D(P)$ is bounded. For closedness consider a sequence $( A_n)_{n\in\mathbb N}$ in $ M_D(P)$ which converges to some \mbox{$ A\in\mathbb C^{n\times n}$}. Thus there exists a sequence $(X_n)_{n\in\mathbb N}$ in $P$ and a sequence of quantum channels $(T_n)_{n\in\mathbb N}$ such that $T_n(D)=D$ and $T_n(X_n)= A_n$. As $P$ is assumed to be compact there exists a subsequence $(X_{n_j})_{j\in\mathbb N}$ of $(X_n)_{n\in\mathbb N}$ which converges to some $X\in P$. On the other hand, by Lemma \ref{lemma_convex_subsemigroup} there also exists a subsequence $(T_{n_l})_{l\in\mathbb N}$ of $(T_{n_j})_{j\in\mathbb N}$ which converges to some $T\in Q(n)$. Combining these two yields subsequences $(X_{n_l})_{l\in\mathbb N},(T_{n_l})_{l\in\mathbb N}$ with coinciding index set which satisfy
\begin{align*}
\|T(X)-T_{n_l}(X_{n_l})\|_1&\leq \|T(X)-T_{n_l}(X)\|_1+\|T_{n_l}(X)-T_{n_l}(X_{n_l})\|_1\\
&\leq \|T-T_{n_l}\|\|X\|_1+\|X-X_{n_l}\|_1\to 0\quad\text{ as }l\to\infty\,.
\end{align*}
Therefore $T(D)=\lim_{l\to\infty} T_{n_l}(D)=D$ and
$$
 A=\lim_{l\to\infty} A_{n_l}=\lim_{l\to\infty}T_{n_l}(X_{n_l})=T(X)\,,
$$
so $ A\in M_D(P)$ as $X\in P$. (iv): Following Remark \ref{rem_semigroup} \mbox{$M_D(B)= \operatorname{conv}\lbrace T(B)\,|\, T\in Q^E_D(n) \rbrace$} so the statement in question follows from Minkowski's theorem \cite[Thm.~5.10]{Brondsted83}, that is, the extreme points of $M_D(B)$ have to be contained within $\lbrace T(B)\,|\, T\in Q^E_D(n) \rbrace$.
\end{proof}
To discuss continuity of the map $M_D$ we first need a (relative) topology on the power set $\mathbb C^{n\times n}$---for this we shall consider the \emph{Hausdorff metric} $\Delta$ on the set of all non-empty
compact subsets $\mathcal P_c(X)\subset\mathcal P(X)$ of a metric space $(X,d)$
\footnote{Let $A,B\in \mathcal P_c(X)$ and $d$ being the metric on $X$. The Hausdorff metric is defined via
\begin{align*}
\Delta(A,B) := \max\Big\lbrace \max_{z \in A}d(z,B),\max_{z \in B}d(z,A) \Big\rbrace\,,
\end{align*}
where as usual $d(z,B)=\min_{w\in B}d(z,w)$ (and $d(z,A)$ analogously), refer to, e.g., \cite{Nadler78}. Then for $A,B\in P_c(X)$ and $\gamma\geq 0$ one has $\Delta(A,B) \leq \gamma$ if and only if for all $a \in A$, there exists $b \in B$ 
with $d(a,b) \leq \gamma$ and vice versa, cf. \cite[Lemma 2.3]{DvE18}.
\label{footnote_hausdorff}}.
\begin{proposition}
Let $d\in\mathbb R_{++}^n$. Then the map $M_D:\mathcal P_c(\mathbb C^{n\times n})\to \mathcal P_c(\mathbb C^{n\times n})$ is well-defined and non-expansive, that is,
$$
\Delta( M_D(P_1),M_D(P_2))\leq \Delta(P_1,P_2)
$$
for all $P_1,P_2\in\mathcal P_c(\mathbb C^{n\times n})$ when equipping $\mathbb C^{n\times n}$ with the trace norm. In particular $M_D$ is continuous.
\end{proposition}
\begin{proof}
Well-definedness is due to Theorem \ref{lemma_R_d_closed} (iii). Now $M_D$ is non-expansive as
\begin{align*}
\max_{A_1\in M_D(P_1)}\min_{A_2\in M_D(P_2)}&\|A_1-A_2\|_1=\max_{\substack{T\in Q_D(n)\\B_1\in P_1}}\min_{\substack{S\in Q_D(n)\\B_2\in P_1}}\|T(B_1)-S(B_2)\|_1\\
&\leq \max_{\substack{T\in Q_D(n)\\B_1\in P_1}}\min_{B_2\in P_1}\|T(B_1)-T(B_2)\|_1\\
&\leq \max_{\substack{T\in Q_D(n)\\B_1\in P_1}}\min_{B_2\in P_1}\|T\|\|B_1-B_2\|_1=\max_{B_1\in P_1}\min_{B_2\in P_1}\|B_1-B_2\|_1
\end{align*}
where in the second-to-last step we used Lemma \ref{lemma_qc_op_bound}.
\end{proof}

Finally, one finds the somewhat peculiar property that applying $M_D$ as well as $M_\identity$ (that is, classical matrix majorization) alternately to some initial states then one, in the closure, ends up with all states:
\begin{proposition}\label{lemma_3}
Let $d\in\mathbb R_{++}^n$ such that $d$ and $\unitvector$ are linearly independent, i.e.~$d\neq c\unitvector$ for all $c\in\mathbb R_{++}$. Then for arbitrary $\rho\in\mathbb D(\mathbb C^n)$ 
$$
\lim_{m\to\infty}(M_{\identity}\circ M_D)^m(\rho)= \mathbb D(\mathbb C^n)
$$
with respect to the Hausdorff metric (cf.~footnote \ref{footnote_hausdorff}).
\end{proposition}
\begin{proof}
The case $n=1$ is obvious so consider $n>1$. Also w.l.o.g.~we may assume that $\unitvector^Td=1$---else we can rescale the problem accordingly. First be aware that applying the Hausdorff metric is allowed due to the following facts:
\begin{itemize}
\item $\mathbb D(\mathbb C^n)$ is compact (Lemma \ref{lemma_pure_extreme})\smallskip
\item $M_D$ for any $D>0$ maps non-empty compact sets to non-empty compact sets (Theorem \ref{lemma_R_d_closed}) 
$(M_{\identity}\circ M_D)^m(\rho)\subseteq \mathbb D(\mathbb C^n)$ itself is compact for all $m\in\mathbb N_0$.
\end{itemize}
Because $M$ is a closure operator---so in particular it is extensive---for any compact set $P\subseteq \mathbb D(\mathbb C^n)$ the sequence 
$\big((M_{\identity}\circ M_D)^m(P)\big)_{m\in\mathbb N}\subseteq \mathbb D(\mathbb C^n)$ 
is increasing with respect to $\subseteq$. Therefore \cite{Baronti86} implies that the sequence converges with respect to the Hausdorff metric $\Delta$ with compact limit set $\overline{\bigcup_{m=1}^\infty (M_{\identity}\circ M_D)^m(P)}\subseteq\mathbb D(\mathbb C^n)$.

The idea will be the following: first we show by explicit construction that starting from $D=\diag(d)\in\mathbb D(\mathbb C^n)$ we can approximately reach $ |e_1\rangle\langle e_1|$. We then may use extensiveness as well as continuity of $M$ w.r.t.~$\Delta$ 
to get\footnote{
Note that in the second-to-last row of the following 
computation we will make use of the following basic result (cf., e.g., \cite[Lemma 2.5.(a)]{DvE18}): Let $(A_n)_{n\in\mathbb N}$ and 
$(B_n)_{n\in\mathbb N}$ be bounded sequences of non-empty compact subsets
of any metric space which Hausdorff-converge to $A$ and $B$, respectively. Now if $A_n\subseteq B_n$ for all $n\in\mathbb N$, then $A \subseteq B$.
}
\begin{align*}
\mathbb D(\mathbb C^n)\supseteq \lim_{m\to\infty}(M_{\identity}\circ M_D)^m(\rho)&=\lim_{m\to\infty}\big(M_{\identity}\circ M_D\circ(M_{\identity}\circ M_D)^{m-2}\circ M_{\identity}\circ M_D\big)(\rho)\\
&= (M_{\identity}\circ M_D)\big(\lim_{\tilde m\to\infty} (M_{\identity}\circ M_D)^{\tilde m}( M_{\identity}\circ M_D)(\rho)\big)\\
&\supseteq (M_{\identity}\circ M_D)\big(\lim_{\tilde m\to\infty} (M_{\identity}\circ M_D)^{\tilde m}(D)\big)\\
&\supseteq (M_{\identity}\circ M_D)( |e_1\rangle\langle e_1|)\supseteq M_{\identity} ( |e_1\rangle\langle e_1|)=\mathbb D(\mathbb C^n)
\end{align*}
because $D$ is minimal in $\mathbb D(\mathbb C^n)$ w.r.t.~$\prec_D$ and 
every pure state is maximal in $\mathbb D(\mathbb C^n)$ w.r.t.~$\prec_{\identity}$ (that is, $\prec$) by Lemma \ref{d_matrix_minmax}. This would conclude the proof.

Carrying out this idea, by assumption we find $j\in\lbrace 1,\ldots,n-1\rbrace$ such that $d_j\neq d_{j+1}$ (w.l.o.g.~$d_j> d_{j+1}$, the other case is shown analogously). Define
\begin{align*}
T:=\begin{pmatrix} \identity_{j-1}&0&0&0\\0&1-\frac{d_{j+1}}{d_j}&1&0\\0&\frac{d_{j+1}}{d_j}&0&0\\0&0&0&\identity_{n-j-1} \end{pmatrix}\in s_d(n)
\end{align*}
and let $\sigma_r=\sum_{i=1}^n |e_{i+1}\rangle\langle e_i|\in s_{\unitvector}(n)$ (with $e_{n+1}:=e_1$) be the cyclic right shift. Starting from any $x\in\mathbb R_{++}^n$ with $\unitvector^Tx=1$ one computes
\begin{equation}\label{eq:x_1_rec}
x^{(1)}:=\sigma_r^{n-j+1}T\prod_{k=1}^{n-2}(\sigma_rT)\sigma_r^j\identity x=\begin{pmatrix} 1-\frac{d_{j+1}}{d_j}(1-x_1)\\\frac{d_{j+1}}{d_j}x_2\\\cdots\\\frac{d_{j+1}}{d_j}x_n \end{pmatrix}=\Big(1-\frac{d_{j+1}}{d_j}\Big)e_1+\frac{d_{j+1}}{d_j}x
\end{equation}
where $x^{(1)}\in\mathbb R_{++}^n$ and $\diag{x^{(1)}}\in (M_{\identity}\circ M_D)^n(\diag x)$. Applying this step successively $\alpha\in\mathbb N$ times results in
$$
x^{(\alpha+1)}:=\Big(\sigma_r^{n-j+1}T\prod_{k=1}^{n-2}(\sigma_rT)\sigma_r^j\identity\Big) x^{(\alpha)}=\Big(1-\Big(\frac{d_{j+1}}{d_j}\Big)^{\alpha+1}\Big)e_1+\Big(\frac{d_{j+1}}{d_j}\Big)^{\alpha+1}x
$$
as is evident by induction invoking \eqref{eq:x_1_rec}. Due to $\diag{x^{(\alpha)}}\in (M_{\identity}\circ M_D)^{n\alpha}(\diag x)$ for all $\alpha$ we found a sequence in $\big((M_{\identity}\circ M_D)^{n\alpha}(\diag x)\big)_{\alpha\in\mathbb N}$ which converges to $ |e_1\rangle\langle e_1|$:
\begin{align*}
\big\| |e_1\rangle\langle e_1|-\diag{x^{(\alpha)}}\big\|_1&=\Big\| |e_1\rangle\langle e_1|-\Big(1-\Big(\frac{d_{j+1}}{d_j}\Big)^{\alpha}\Big) |e_1\rangle\langle e_1|-\Big(\frac{d_{j+1}}{d_j}\Big)^{\alpha}\diag x\Big\|_1\\
&= \Big(\frac{d_{j+1}}{d_j}\Big)^{\alpha}\| |e_1\rangle\langle e_1|-\diag x\|_1\overset{\alpha\to\infty}\longrightarrow 0
\end{align*}
Thus by the limit point characterization of Hausdorff convergence\footnote{
For bounded sequences $(A_n)_{n\in\mathbb N}$ of non-empty
compact subsets (of some metric space) which Hausdorff-converges to $A$, one has $x\in A$ if and only if there exists a sequence $(a_n)_{n\in\mathbb N}$ with 
$a_n \in A_n$ which converges to $x$, cf.~\cite[Lemma 2.4]{DvE18}.
}
\begin{equation}
|e_1\rangle\langle e_1|\in \lim_{\alpha\to\infty} (M_{\identity}\circ M_D)^{n\alpha}(D)=\lim_{\tilde m\to\infty} (M_{\identity}\circ M_D)^{\tilde m}(D)\,.\tag*{\qedhere}
\end{equation}
\end{proof}
\noindent This result is non-trivial in the following sense: if the initial $\rho$ is of full rank then by Corollary \ref{thm_full_rank} $(M_1\circ M_D)^m(\rho)$ for arbitrary $m$ can never equal all of $\mathbb D(\mathbb C^n)$ but is only a proper subset.

\section{Conclusion and Outlook}
While strict positivity ({\sc sp}) forms a rather large subclass of usual positivity ({\sc p}) it turns out to be an inherent feature of Markovian quantum processes---the main reason for this being that {\sc sp} forms a convex semigroup which is open with respect to {\sc p}. Now attempting to take this to infinite-dimensional Hilbert spaces creates some problems---for example the proof of Proposition \ref{prop_positivity} breaks down as not every positive semi-definite operator has a smallest non-zero eigenvalue. Although one might fix this by readjusting the definition positive definiteness to operators which are bounded from below (that is, invertible) this would be problematic from a physical point of view because channels in the Schr{\"o}dinger picture act on trace class operators---the eigenvalues of which can never be bounded away from zero. This might also show when taking Proposition \ref{mainthm} (iv), or Remark \ref{rem_ch_heisenberg} depending on the physical picture, as a definition for infinite-dimensional systems as then one could choose whether the orthogonal projections $\pi$ have to be finite-dimensional or can be arbitrary. 

Either way strict positivity proves to be a useful tool in analyzing $D$-majorization on matrices and its order properties. The biggest difference between the vector and the general matrix case arguably lies within the convex compact set of all matrices $D$-majorized by some initial matrix from Eq.~\eqref{eq:def_MD}: in the vector case this set has finitely many extreme points meaning it forms a convex polytope. This grants access to strong tools from the underlying mathematical theory and allows for an analytical extreme point analysis. In the matrix case, however, the number of extreme points is infinite\footnote{
Similarly while the set of $d$-stochastic matrices forms a convex polytope, the set $Q_D(n)$ has infinitely many extreme points. The argument, just as below, relies on the concatenation with suitable unitary channels from left and right---after all the bijective quantum channels the inverse of which is a channel again are precisely the unitary ones \cite[Prop.~1]{DvE18}.
}---a straightforward argument shows that if $X$ is extremal in $ M_D(A)$ then so is $U^\conj XU$ for all unitaries $U$ which satisfy $[U,D]=0$. A way of bridging this gap could be to define the unitary equivalence relation
$$
\sim\;:=\{(X,UXU^*)\,|\,X\in M_D(A), U\in\mathbb C^{n\times n}\text{ unitary with }[U,D]=0\}
$$
for arbitrary $d\in\mathbb R_{++}^n$ and $A\in\mathbb C^{n\times n}$, and then look at the equivalence class of an extreme point $X$ of $M_D(A)$ under $\sim$. Whether one ends up with finitely many extreme points after factoring out the unitary equivalents for now remains open.

%
Other than simplifying extreme point analysis answering this question might also enable a proof for continuity of the map $(D,P)\mapsto M_D(P)$ defined for $D>0$ and $P\in\mathcal P_c(\mathbb C^{n\times n})$---the latter possibly restricted to just the hermitian matrices. In the vector case this map is indeed continuous \cite[Thm.~4.8]{vomEnde19polytope} but the proof relies heavily on the polytope structure and the corresponding half-space representation.



\bigskip\smallskip
\noindent\textbf{Acknowledgments.} I would like to thank Gunther Dirr, 
Thomas Schulte-Herbr\"uggen, Michael M.~Wolf, Michael Keyl as well as the anonymous referee for valuable and constructive comments. Also this manuscript greatly benefited from my time in Toru{\'n}, in particular from my illuminating discussions there with Sagnik Chakraborty, Ujan Chakraborty and Dariusz Chru{\'s}ci{\'n}ski. This work was supported 
by the Bavarian excellence network {\sc enb}
via the \mbox{International} PhD Programme of Excellence
{\em Exploring Quantum Matter} ({\sc exqm}).

\bibliographystyle{tfnlm}

\bibliography{../../../../control21vJan20}

\section{Appendix}\label{appendix}
\renewcommand{\thesubsection}{\Alph{subsection}}

\subsection{Proof of Proposition \ref{prop_exist_channel}}\label{app_a0}

\begin{lemma}\label{lemma_vector_transform}
Let $n\in\mathbb N\setminus\{1\}$, $x\in\mathbb R^n$. Then there exists a column-stochastic matrix $A\in\mathbb R^{n\times n}$ such that
$
Ax=( \unitvector^Tx_+,-\unitvector^Tx_-,0,\ldots,0)^T
$
where $x=x_+-x_-$ is the unique decomposition of $x$ into positive and negative part, i.e.~$x_+,x_-\in\mathbb R^n$, $\langle x_+,x_-\rangle=0$. 
\end{lemma}
\begin{proof}
For arbitrary $j=1,\ldots,n$ choose $A_{1j}=1, A_{2j}=0$ if $x_j\geq 0$ and \mbox{$A_{1j}=0$}, $A_{2j}=1$ if $x_j< 0$ as well as $A_{ij}=0$ for all $i>2$. Then every column contains precisely one $1$-entry and the rest is $0$, hence $A_1$ is column-stochastic and
\begin{equation}
(A_1x)_1=\sum\nolimits_{x_j\geq 0} x_j=\sum\nolimits_{j=1}^n \max\{x_j,0\}=\unitvector^Tx_+ \qquad (A_1x)_2=-\unitvector^Tx_-\,.\tag*{\qedhere}
\end{equation}
\end{proof}

\begin{lemma}\label{lemma_stoch_transf}
For $x,y\in\mathbb R^n$ the following statements are equivalent.
\begin{itemize}
\item[(i)] $\unitvector^Tx=\unitvector^Ty$ and $\|x\|_1\leq\|y\|_1$
.\smallskip
\item[(ii)] There exists a column-stochastic matrix $M\in\mathbb R^{n\times n}$ such that $My=x$.
\end{itemize}
\end{lemma}
\begin{proof}
(ii) $\Rightarrow$ (i): Simple calculation involving the triangle inequality.

(i) $\Rightarrow$ (ii): By Lemma \ref{lemma_vector_transform} there exists $A\in\mathbb R^{n\times n}$ column-stochastic such that \mbox{$Ay=( \unitvector^Ty_+,-\unitvector^Ty_-,0,\ldots,0)^T$}. Also $\|x\|_1\leq\|y\|_1$ implies that $x\prec(\unitvector^T y_+,-\unitvector^T y_-,0,\ldots,0)^T$ so we find $B\in\mathbb R^{n\times n}$ doubly stochastic with $B(\unitvector^Ty_+,-\unitvector^Ty_-,0,\ldots,0)^T=x$ \cite[Ch.~2, Thm.~B.6]{MarshallOlkin}. Thus $M:=BA$ is column-stochastic and satisfies $My=B(Ay)=x$.
\end{proof}

This enables proving the result in question.
\begin{proof}[Proof of Proposition \ref{prop_exist_channel}]
``(ii) $\Rightarrow$ (iii)'': Obvious. ``(iii) $\Rightarrow$ (i)'': Lemma \ref{lemma_qc_op_bound}.

``(i) $\Rightarrow$ (ii)'': There exist unitaries $U,V\in\mathbb C^{n\times n}$ and vectors $x,y\in\mathbb R^n$ such that $A=U\operatorname{diag}(x) U^\conj$, $B=V\operatorname{diag}(y) V^\conj$. By assumption $\unitvector^Tx=\tr(A)=\tr(B)=\unitvector^Ty$ and $\|x\|_1=\|A\|_1\leq\|B\|_1=\|y\|_1$. Hence Lemma \ref{lemma_stoch_transf} yields a column-stochastic matrix $M\in\mathbb R^{n\times n}$ with $My=x$. Define a map $\tilde{T}:\mathbb C^{n\times n}\to \mathbb C^{n\times n}$ via
\begin{align*}
|e_i\rangle\langle e_j|\mapsto\begin{cases}
0&\text{ if }i\neq j\\
\sum\nolimits_{k=1}^n M_{ki} |e_k\rangle\langle e_k|&\text{ if }i= j
\end{cases}
\end{align*}
and its linear extension onto all of $\mathbb C^{n\times n}$. The Choi matrix of $\tilde{T}$ is diagonal with non-negative entries because $M_{jk}\geq 0$ for all $j,k$ so $C(\tilde{T})\geq 0$ and $\tilde{T}$ is completely positive by Lemma \ref{lemma_choi_matrix}. Moreover $\tilde{T}$ is trace preserving because
$$
\tr\big(\tilde{T}( |e_j\rangle\langle e_j|)\big)=\sum\nolimits_{k=1}^n M_{kj}\tr( |e_k\rangle\langle e_k|)=(\unitvector^TM)_j=1=\tr( |e_j\rangle\langle e_j|)
$$
for all $j=1,\ldots,n$. This shows $\tilde{T}\in Q(n)$. Also 
\begin{align*}
\tilde{T}(\operatorname{diag}(y))=\sum\nolimits_{j=1}^n y_j \tilde{T}( |e_j\rangle\langle e_j|)&=\sum\nolimits_{i=1}^n\Big(\sum\nolimits_{j=1}^nM_{ij}y_{j}\Big) |e_i\rangle\langle e_i|\\
&=\sum\nolimits_{i=1}^n(My)_i |e_i\rangle\langle e_i|=\sum\nolimits_{i=1}^nx_i |e_i\rangle\langle e_i|=\operatorname{diag}(x)
\end{align*}
so $T(\cdot):=U \tilde T(V^\conj (\cdot)V)U^\conj$ ($\in Q(n)$ as a composition of quantum channels, Lemma \ref{lemma_convex_subsemigroup}) satisfies $T(B)=A$. 
Now if one of the $y_j$ (eigenvalues of $B$) is $0$ then the action of $\tilde T( |e_j\rangle\langle e_j|)=:\omega$ can obviously be chosen freely without affecting \mbox{$\tilde T(\diag y)=\diag x$}, that is, $T(B)=A$. If $\omega\in\mathbb D(\mathbb C^n)$ then $\tilde T,T$ remain in $Q(n)$ by the above argument so defining $\psi:=Ve_j$ concludes the proof. 
\end{proof}
\subsection{Examples}
\begin{example}\label{example_1}
The linear map
\begin{align*}
T:\mathbb C^{2\times 2}&\to\mathbb C^{2\times 2}\\
 \begin{pmatrix} a_{11}&a_{12}\\a_{21}&a_{22}\end{pmatrix}&\mapsto\begin{pmatrix} a_{11}+\frac12a_{22}&0\\0&\frac12a_{22}\end{pmatrix}
\end{align*}
is obviously {\sc cptp} and strictly positive ($T(\identity)>0$) but the only fixed points of $T$ are of the form $\footnotesize\begin{pmatrix}x&0\\0&0 \end{pmatrix}$, that is, not of full rank.
\end{example}
\begin{example}\label{example_2}
Consider the channel
\begin{align*}
T:\mathbb C^{3\times 3}&\to \mathbb C^{3\times 3}\\
\begin{pmatrix} a_{11}&a_{12}&a_{13}\\a_{21}&a_{22}&a_{23}\\a_{31}&a_{32}&a_{33} \end{pmatrix}&\mapsto \begin{pmatrix}a_{22}+a_{33}&0&0\\0& \frac12a_{11}&0\\0&0& \frac12a_{11} \end{pmatrix}\,.
\end{align*}
In particular this map is strictly positive by Prop.~\ref{prop_positivity} (iii) as $\operatorname{diag}(2,1,1)$ is a fixed point. However
\begin{align*}
1=\rank{|e_1\rangle\langle e_1|}&<\rank{T(|e_1\rangle\langle e_1|)}=2\\
2=\rank{|e_2\rangle\langle e_2|+|e_3\rangle\langle e_3|}&>\rank{T(|e_2\rangle\langle e_2|+|e_3\rangle\langle e_3|)}=1\,.
\end{align*}
\end{example}

\begin{example}\label{example_3new}
For any $m\in\mathbb N$ define $T_m:\mathbb C^{2\times 2}\to \mathbb C^{2\times 2}$ via
$$
T_m\begin{pmatrix} a_{11}&a_{12}\\a_{21}&a_{22}\end{pmatrix}=\begin{pmatrix} (1+\frac1m)a_{11}-\frac1ma_{22}&a_{12}\\a_{21}&(1+\frac1m)a_{22}-\frac1ma_{11} \end{pmatrix}\,.
$$
Obviously every $T_m$ is trace-preserving but not positive as
$$
T_m\begin{pmatrix} 1&0\\0&0 \end{pmatrix}=\begin{pmatrix} 1+\frac1m&0\\0&-\frac1m \end{pmatrix}\,.
$$
However $\lim_{m\to\infty}T_m=\operatorname{id}_2$ so for every $\varepsilon>0$ there exists $T\in B_\varepsilon(\operatorname{id}_2)$ which is not strictly positive, although the identity itself is strictly positive. 
This example can easily be generalized to arbitrary sizes of domain and co-domain. 
\end{example}

\begin{example}\label{example_5}
\begin{itemize}
\item[(i)] The Choi matrix of the linear map
\begin{align*}
T:\mathbb C^{3\times 3}&\to \mathbb C^{3\times 3}\\
\begin{pmatrix} a_{11}&a_{12}&a_{13}\\a_{21}&a_{22}&a_{23}\\a_{31}&a_{32}&a_{33} \end{pmatrix}&\mapsto \begin{pmatrix} a_{11}&\frac{i}{\sqrt2}(a_{12}+a_{13})&0\\-\frac{i}{\sqrt2}(a_{21}+a_{31})&a_{22}+a_{33}&0\\0&0&0 \end{pmatrix}
\end{align*}
has simple eigenvalues $2,1$ and the $7$-fold eigenvalue $0$ so $T$ is {\sc cptp}, not {\sc sp} and not a trace projection, that is, not of the form $A\mapsto\tr(A)\rho$ for any state $\rho$.\smallskip
\item[(ii)] Via $\tr(T(A)B)=\tr(AT^*(B))$ for all $A,B\in\mathbb C^{3\times 3}$ the dual of $T$ from (i) is
\begin{align*}
T^*:\mathbb C^{3\times 3}&\to \mathbb C^{3\times 3}\\
\begin{pmatrix} b_{11}&b_{12}&b_{13}\\b_{21}&b_{22}&b_{23}\\b_{31}&b_{32}&b_{33} \end{pmatrix}&\mapsto \begin{pmatrix} b_{11} & -\frac{i}{\sqrt{2}}b_{12} & -\frac{i}{\sqrt{2}}b_{12} \\ \frac{i}{\sqrt{2}}b_{21} & b_{22} & 0 \\ \frac{i}{\sqrt{2}}b_{21} & 0 & b_{22} \end{pmatrix}\,.
\end{align*}
Note that  the action of $T^*$ is determined by a subalgebra of the domain because 
$$
T^*(B)=T^*\begin{pmatrix} b_{11}&b_{12}&0\\b_{21}&b_{22}&0\\0&0&0 \end{pmatrix}\,.
$$
\end{itemize}
\end{example}
\begin{example}\label{example_5c}
Consider the unitary matrix $\sigma={\footnotesize\begin{pmatrix} 0&1\\1&0 \end{pmatrix}}$ and the induced channel $T:\mathbb C^{2\times 2}\to \mathbb C^{2\times 2}$, $\rho\mapsto \sigma\rho\sigma$. Then
$$
\|T-\id\|\geq \big\|T\big(|e_1\rangle\langle e_1|\big)-|e_1\rangle\langle e_1|\,\big\|_1=\big\|\,|e_2\rangle\langle e_2|-|e_1\rangle\langle e_1|\,\big\|_1=2
$$
but as a unitary channel, $T$ preserves the identity and thus is strictly positive. 
\end{example}
\end{document}